\documentclass[final,leqno]{siamltex}

\usepackage{showkeys}

\usepackage{latexsym,calc,amssymb,amsfonts}
\usepackage{amsmath}
\usepackage{epsfig,setspace,color}

\newtheorem{remark}{Remark}

\def\d{\ensuremath{\delta}}
\def\r{\ensuremath{\rho}}
\def\e{\ensuremath{\epsilon}}
\def\l{\ensuremath{\lambda}}
\def\G{\ensuremath{\Gamma}}

\def\RR{\mathbb{R}}

\newcommand{\al}{A_{K}} %This way we don't have to go through and change the shortcuts individually
\newcommand{\ak}{A_K} % but this is more intuitive for the revision

\newcommand{\wnk}{W_{n,k}} % for wishart matrices

\def\pdfmin{f_{min}(k,n;\lambda)}
\def\pdfmax{f_{max}(k,n;\lambda)}

\def\Lmax{\Lambda^{max}_{n,k}}
\def\lmaxW{\lambda^{max}(W_{n,k})}
\def\Lmin{\Lambda^{min}_{n,k}}
\def\lminW{\lambda^{min}(W_{n,k})}

\def\lmax{\lambda^{max}}
\def\lmin{\lambda^{min}}

\def\kn{\frac{k}{n}}
\def\nN{\frac{n}{N}}

\def\lone{\ell^1}

\def\rd{\rho_S^D(\delta)}
\def\rrv{\rho_S^{RV}(\delta)}

\def\rflq{\rho_S^{FL}(\delta;q)}
\def\rflqcond{\rho_S^{FL}(\delta;q,Cond)}
\def\rflqc1{\rho_S^{FL}(\delta;q,C_1(\delta,\rho)\le\Upsilon)}

\def\rfl{\rho_S^{FL}(\delta)}

\def\xs{x^{\star}_{\theta}}
\def\xt{x_{\theta}}

\def\gmax{g_{max}(k,n;\lambda)}

\def\rn{\rho_n}
\def\dn{\delta_n}

\def\sl1{StrongEquiv(A,\ell^1)}

\title{Compressed Sensing: \\ How sharp is the Restricted Isometry Property?}

\author{Jeffrey D. Blanchard\thanks{Department of Mathematics and Statistics, Grinnell College, Grinnell, Iowa, USA ({\tt jeff@math.utah.edu}).  JDB acknowledges support from NSF DMS (VIGRE) grant number 0602219 while a VIGRE postdoctoral fellow at Department of Mathematics, University of Utah.}
      \and Coralia Cartis\thanks{School of Mathematics, University of Edinburgh, Edinburgh, UK ({\tt coralia.cartis@ed.ac.uk}).} \and Jared Tanner\thanks{School of Mathematics, University of Edinburgh, Edinburgh, UK ({\tt jared.tanner@ed.ac.uk}).  JT acknowledges support from the Alfred P. Sloan Foundation, Leverhulme Trust,  and thanks John E. and Marva M. Warnock for their generous support in the form of an endowed chair.}}

\date{January 2009}

\begin{document}

\maketitle

\begin{abstract} 
Compressed Sensing (CS) seeks to recover an unknown vector with $N$ entries 
by making far fewer than $N$ measurements; it posits that the number of 
compressed sensing measurements should be comparable to the information 
content of the vector, not simply $N$.  CS combines the important task of compression 
directly with the measurement task.  Since its introduction in 2004 there have been 
hundreds of manuscripts on CS, a large fraction of which develop 
algorithms to recover a signal from its compressed measurements.

Because of the paradoxical nature of CS -- exact reconstruction from seemingly undersampled 
measurements -- it is crucial for acceptance of an algorithm that rigorous analyses verify the 
degree of undersampling the algorithm permits.  The Restricted Isometry Property (RIP) has 
become the dominant tool used for the analysis in such cases.

We present here an asymmetric form of RIP which gives tighter bounds than the usual 
symmetric one.  We give the best known bounds on the RIP constants for matrices from 
the Gaussian ensemble.  Our derivations illustrate the way in which the combinatorial 
nature of CS is controlled.  Our quantitative bounds on the RIP allow precise statements 
as to how aggressively a signal can be undersampled, the essential question for practitioners.
We also document the extent to which RIP gives precise information about the true performance 
limits of CS, by comparing with approaches from high-dimensional geometry.
\end{abstract}


\begin{AMS}
Primary: 41A46, 94A12, 94A20. Secondary: 15A52, 60F10, 90C25.
%in order these classification are: approx. by arbitrary nonlinear expressions, signal theory, sampling theory, random matrices, large deviations, convex programming
\end{AMS}

\pagestyle{myheadings}
\thispagestyle{plain}
\markboth{J. D. BLANCHARD, C. CARTIS, AND J. TANNER}{Compressed Sensing: How sharp is the RIP?}

%%%%%%%%%%% INTRODUCTION %%%%%%%%%%%%%%%%%%%%%%%%%%%

\section{Introduction}\label{sec:Intro}
Consider the task of measuring an unknown vector $x\in\RR^N$ by taking inner products with vectors of one's choosing.  The obvious choice would be to ask for the inner product of $x$ with respect to each of the $N$ canonical unit vectors $e_j$ (the $j^{th}$ entry of $e_j$ being one and all others zero).  But what if it is known a priori that $x$ is $k$-sparse -- i.e. has only $k<N$ nonzero entries?  Can't one then do better?  If the nonzero entries of $x$ are indexed by the set $K$ ($x(j)\ne 0$ if $j\in K$ and $x(j)=0$ for $j\in K^c$), then only $k$ inner products are needed: those with the canonical unit vectors $e_j$ for $j\in K$.  However, what if $K$ is unknown?  Is it still possible to make fewer than $N$ measurements of $x$?  

Questions of this form must have been around for millennia.  Consider this puzzle:  ``A counterfeit coin is hidden in a batch of $N$ otherwise similar coins; it is distinguished from the others by its slightly heavier weight.  How many balance weighings are needed to find the counterfeit?''  Abstractly, this concerns the special case where $K$ is an unknown singleton and the nonzero value is nonnegative; the balance is abstractly the same as an inner product which gives weight +1 to the coefficients placed in the ``right'' pan and -1 to the coefficients placed in the ``left'' pan.  Many people quickly find that roughly $\log(N)$ measurements suffice to find the position and value of the nonzero, each time putting half the remaining coins in one pan, half in the other, and discarding from further consideration the coins that turn out on the light side.  Lighthearted as puzzles can sometimes seem, they can lead to serious applications.

During World War Two, efficient screening of large groups of soldiers for certain infections was based on the principle of group testing, in which blood from many soldiers is combined in a single tube and tested for presence of an infectious agent.  If an infection is found, one studies that group and by dyadic subdivision eventually isolates the infecteds \cite{Dor43,group_cs}.  

More advanced mathematics can do much better than such common-sense ideas.  Those with a physical  bent may quickly see that, if $N$ is prime, again assuming a singleton $K$ and a nonnegative $x$, it will be enough, in fact, to make only 2 inner products, with respectively a sine and cosine of frequency $2\pi/N$; the phase of the corresponding complex Fourier coefficient immediately reveals the position of the nonzero.  Note here that, for large $N$, we are doing dramatically better than common-sense (2 measurements rather than $\log(N)$).

Advanced mathematics is better than the common-sense approach in another way: common-sense uses adaptive measurements, where the next measurement vector is selected after viewing all previous measurements.  In the advanced approach, adaptivity is unnecessary: one simply makes 2 measurements defined a priori and later combines the two to reconstruct.

Compressed Sensing (CS) embodies the advanced approach: it designs a special matrix $A$ of size $n\times N$, measures $x$ via $y=Ax$, giving $n$ measurements of the $N$ vector $x$ in parallel, and reconstructs $x$ from $(y,A)$ using computationally efficient and stable algorithms.  The key point is that $n$ can be taken much smaller than $N$, and much closer to $k$.  For example, if $x$ is known to be $k$-sparse and nonnegative $n=2k+1$ suffices \cite{DoTa05_signal} and if $x$ is only known to be $k$-sparse, roughly $n=2\log(N/n)\cdot k$ will suffice, if $k/N$ is small \cite{DoTa08_JAMS}.

%The affirmative resolution of this question dates back at least to the early 1940s with group testing of large populations for medical diseases with infection rates substantially below one half, \cite{Dor43,group_cs}.  This early resolution is based on adaptive measurements, where the next measurement vector is selected after viewing all previous measurements.  In fact, it has recently been shown that adaptivity is unnecessary.  Namely, a matrix $A$ of size $n\times N$ is fixed and then all $n$ measurements are taken in parallel via $y=Ax$; further, $x$ can be recovered from $y$ and $A$ using computationally efficient and stable algorithms.  This realization is the essential core of the recent topic Compressed Sensing.  Compressed Sensing (CS) proposes that one makes considerably fewer measurements than $N$ by assuming that the signals of interest are compressible.  

Since the release of the seminal CS papers in 2004, \cite{CTDecoding, CRTRobust, CompressedSensing}, a great deal of excitement has been generated in signal processing and applied mathematics research, with hundreds of papers on the theory, applications, and extensions of compressed sensing  (more than 400 of these are collected at Rice's online Compressive Sensing Resources archive {\tt dsp.rice.edu/cs}).  Many applications have been proposed, including magnetic resonance imaging \cite{CS_MRI, mehmet}, radar \cite{Healy_radar}, and single-pixel cameras \cite{CS_pixel} to name a few.  In the MRI 
applications, it has been reported that diagnostic quality images can be obtained in $1/7$ the recording time using CS approaches, \cite{MRI_fast}.  For a recent review of CS see the special issue containing 
\cite{CS_pixel, CS_MRI} and for a review of sparse approximation see \cite{CS_SIREV}.

In CS the matrix $A$ and reconstruction algorithm are referred to as an encoder/decoder pair and much of the research has focused on their construction; that is, how should the measurement matrix $A$ be selected and what are the most computationally efficient and robust algorithms for recovering $x$ given $y$ and $A$?  
The two most prevalent encoders in the literature construct $A$ by drawing its entries independently and identically from a Gaussian normal distribution, or by randomly sampling its rows without replacement from amongst the rows of a Fourier matrix. These enconders are popular as they are amenable to analysis, and they can be viewed as models of matrices
with mean-zero entries and fast matrix-vector products, respectively.
 The most widely-studied decoder has been $\ell^1$-minimization,
\begin{equation}\label{eq:l1}
\min_{z\in\RR^N} \|z\|_{1}\ \hbox{subject to}\ Az=y,
\end{equation}
which is the convex relaxation of the computationally intractable decoder, \cite{NPhard},  seeking the sparsest solution in agreement with the measurements
\begin{equation}\label{eq:l0prime}
\min_{z\in\RR^N} \|z\|_{0}\ \hbox{subject to}\ Az=y.
\end{equation}
Following the usual convention in the CS community, $\|z\|_{0}$ counts the number of nonzero entries in $z$.  Many other encoder/decoder pairs are also being actively studied, with new alternatives being proposed regularly; see Section \ref{sec:Algorithms}.  

Here we do not review these exciting activities, but focus our attention on how to interpret the existing theoretical guarantees; in particular, we believe an important task for theory is to correctly predict the triples $(k,n,N)$ for which a given encoder/decoder will successfully recover the measured signal, or a suitable approximation thereof.  To exemplify this, we restrict our attention to a now-standard encoder/decoder pair: $A$ Gaussian and $\ell^1$-minimization.  This pair offers the cleanest mathematical structure, giving us the chance to make the strongest and clearest statements which can be made at this time, for example by drawing on the existing wealth of knowledge in random matrix theory and high-dimensional convex geometry.  In this paper we focus almost exclusively on the most widely used tool for analyzing the performance of encoder/decoder pairs,  the \emph{Restricted Isometry Property} (RIP) introduced by Cand\`es and Tao \cite{CTNearOptimal}. 
\begin{definition}[Restricted Isometry Property]\label{def:rip}
A matrix $A$ of size $n\times N$ is said to satisfy the RIP  with
RIP constant $R(k,n,N; A)$ if, for every $x\in\chi^N(k):=\{x\in\RR^N: \|x\|_0 \le k\}$,
\begin{equation}\label{eq:rip}
R(k,n,N; A):=\min_{c\ge0}\,c\,\,\,\mbox{subject to}\,\,\,
(1-c)\|x\|_2^2 \le \|Ax\|_2^2\le (1+c) \|x\|_2^2.
\end{equation}
\end{definition}
As suggested by the name, the RIP constants measure how much
the matrix $A$ acts like an isometry when ``restricted'' to $k$ columns; it 
describes the most significant distortions of the $\ell^2$ norm of 
any $k$-sparse vector.  Typically, $R(k,n,N;A)$ is 
measured for matrices with unit $\ell^2$-norm columns, and in this special case $R(1,n,N)=0$.
Specifically, the RIP constant
$R(k,n,N;A)$ is the maximum distance from 1 of all the eigenvalues of the
${N\choose k}$ submatrices, $A^T_KA_K$, derived from $A$, where $K$ is an
index set of cardinality $k$ which restricts $A$ to those columns indexed by
$K$.

It is important to note that the RIP is predominently used to establish
theoretical performance guarantees when either the measurement vector $y$ is corrupted with
noise or the vector $x$ is not strictly $k$-sparse.  Proving that an algorithm is stable to noisy measurements is
essential for applications since measurements are rarely free from noise.  In this paper, we focus
on the ideal noiseless case with the hopes of investigating the best
possible theoretical results.  For the noisy case, see \cite{BCT_RIP_arxiv} for
$\ell^q$-minimization for $q\in (0,1]$ and \cite{BCTT09} for greedy algorithms.

For many CS encoder/decoder pairs it has been shown that if the RIP constants for the encoder remain bounded as $n$ and $N$ increase with $n/N\rightarrow\delta\in (0,1)$, then the decoder can be guaranteed to recover the sparsest $x$ for $k$ up to a critical threshold, which can be expressed as a fraction of $n$, $\r(\delta)\cdot n$.  Typically each encoder/decoder pair has a different $\r(\d)$.  Little is generally known about the magnitude of $\r(\delta)$ for encoder/decoder pairs, making it difficult for a practitioner to know how aggressively they may undersample, or which decoder has stronger performance guarantees.  (For a recent review of 
compressed sensing algorithms, including which have $\r(\d)>0$, see \cite[Section 7]{NeTr09_cosamp}.)  In this paper, we endeavor to be as precise as possible about the value of the RIP constants for the Gaussian ensemble, and show how this gives quantitative values for $\r(\d)$ for the $\ell^1$-minimization decoder.  Similar results for other decoders are available in \cite{BCTT09}.

To quantify the sparsity/undersampling trade off, we adopt a \emph{proportional-growth} asymptotic, in which we consider sequences of triples $(k,n,N)$ where all elements grow large in a coordinated way, $n\sim \d N$ and $k\sim \r n$ for some constants $\d,\r>0$.  This defines a two-dimensional phase space $(\d,\r)$ in $[0,1]^2$ for asymptotic analysis.  

\begin{definition}[Proportional-Growth Asymptotic]\label{def:pga}
A sequence of problem sizes $(k,n,N)$ is said to \emph{grow proportionally} if, for $(\d,\r)\in[0,1]^2$,  $\nN\rightarrow\d$ and $\kn\rightarrow\r$
as $n\rightarrow\infty$.
\end{definition}

Ultimately, we want to determine, as precisely as possible, which subset of this phase space corresponds to successful recovery and which subset corresponds to unsuccessful recovery.  This is the phase-transition framework advocated by Donoho et al \cite{Do05_signal,Do05_polytope,DoSt06_breakdown,DoTa05_signal,DoTs06_fast};
see Section \ref{sec:Algorithms} for a precise definition.
By translating the sufficient RIP conditions into the proportional-growth asymptotic, we find lower bounds on the phase-transition for $(\d,\r)$ in $[0,1]^2$.  An answer to this question plays the role of an \emph{undersampling theorem}: to what degree can we undersample a signal and still be able to reconstruct it?  

The central aims of this paper are:
\begin{itemize}
\item to shed some light on the behavior of the RIP constants of a matrix ensemble with as much precision as possible;
\item to advocate a unifying framework for the comparison of theoretical CS results by showing the reader how to interpret and compare some of the existing recovery guarantees for the prevalent $\ell^1$ decoder;
\item to introduce a reader new to this topic to the type of large deviation analysis calculations often encountered in CS and applicable to many areas faced with combinatorial challenges.
\end{itemize}
In pursuit of these goals, we sharpen the use of the RIP and squeeze the most out of it, quantifying what can currently be said in the proportional-growth asymptotic and thereby making precise the undersampling theorems the RIP implies. We proceed in Section~\ref{sec:RIPBounds} along two main avenues.  First, we concentrate on Gaussian matrices, using bounds on their singular values we develop  the sharpest known bounds on their RIP constants; in fact, these are the the best known bounds of any class of matrices in the proportional-growth asymptotic with $n<N$.  Second, we use an asymmetric definition of the RIP where the lower and upper eigenvalues are treated separately, and in doing so further improve the conditions in which the RIP implies CS decoders recover the measured signal.  In Section~\ref{sec:Algorithms} we combine these two improvements to exhibit a region of the $(\d,\r)$ phase space where RIP analysis shows that undersampling will be successful for the $\ell^1$-minimization decoder \eqref{eq:l1}.

The RIP is not the only tool used to analyze the performance of CS decoders. The different methods of analysis lead to results that are rather difficult to compare.  In Section~\ref{subsec:comparison}, we describe in the proportional-growth asymptotic, with $A$ Gaussian and the $\lone$-minimization decoder,  two alternative methods bounding the phase transition: the polytope analysis \cite{Do05_signal, Do05_polytope,DoTa08_JAMS} of Donoho and Tanner and the geometric functional analysis techniques of Rudelson and Vershynin \cite{RV07_gaussian}.  By translating these two methods of analysis and the RIP analysis into the proportional-growth asymptotic, we can readily compare the results obtained by these three techniques by comparing the regions of the $(\d,\r)$ phase space where each method of analysis has guaranteed successful recovery.  In particular, we find that for the Gaussian encoder, the RIP, 
despite its popularity, is currently dramatically weaker than the other two approaches in the strength of conclusions that it can offer.  However, this limitation is counterbalanced by RIP being successfully applied to a broad class of encoder/decoder pairs, and seemlessly also proving stability to noisy measurements and compressible signals.

We conclude with a discussion of some  other important and related topics not addressed in the current paper.  We briefly discuss comparisons of results when noise is present in the measurements or the signal $x$ is not perfectly $k$-sparse, average case analysis versus the theoretical worst case analysis presented here, and the potential to improve the phase transition curves through improved analysis or improved bounds.

%%%%%%%%%%%%% RIP BOUNDS %%%%%%%%%%%%%%%%%%%%%%%%%%%%%%%%%

\section{Bounds on RIP for Gaussian Random Matrices}\label{sec:RIPBounds}

Let $K\subset\left\{1,\dots,N\right\}$ be an index set of cardinality 
$k$ which specifies the columns of $A$ chosen for a submatrix, $\al$, 
of size $n\times k$.   Explicitly computing $R(k,n,N;A)$ would require enumerating all ${N\choose k}$ subsets $K$ of the columns of $A$, forming each  matrix $G_K=\ak^T\ak$, and calculating their largest and smallest eigenvalues.  We have never seen this done except for small sizes of $N$ and $k$, so not much is known about the RIP constants of deterministic matrices.  Fortunately, analysis can penetrate where computation becomes intractable.  Associated with a random matrix ensemble is an, as of yet unknown, probability density function for $R(k,n,N)$.  Let us focus on the Gaussian ensemble where much is already known about its eigenvalues.  We say that an $n\times N$ random matrix $A$ is drawn from the \emph{Gaussian ensemble} of random matrices if the entries are sampled independently and identically from the standard normal distribution, $\mathcal{N}(0,n^{-1})$.  (The $n^{-1}$ scaling in the Gaussian ensemble cause the $\ell^2$ norm of its columns to have expectation 1.)  We say that a $k\times k$ matrix $\wnk$ is a \emph{Wishart} matrix if it is the Gram matrix $X^TX$ of an $n\times k$ matrix $X$ from the Gaussian ensemble.  The largest and smallest eigenvalues of a Wishart matrix are random variables, denoted here $\Lmax=\lmaxW$ and $\Lmin=\lminW$.  These random variables tend to defined limits, in expectation, as $n$ and $k$ increase in a proportional manner.  With $\kn\rightarrow\r$ as $n\rightarrow\infty$, we have ${\cal E}(\Lmax)\rightarrow(1+\sqrt{\r})^2$ and ${\cal E}(\Lmin)\rightarrow(1-\sqrt{\r})^2$; \cite{Geman1980, Silverstein1985}, see Figure~\ref{fig:2deigenvalues}. Explicit formulas bounding $\Lmax$ and $\Lmin$ are available \cite{Edelman_acta}.  An empirical approximation of the probability density functions of $\Lmax$ and $\Lmin$ is shown in Figure \ref{fig:pdf}.

The asymmetric way that the expected eigenvalues $\Lmax$ and $\Lmin$ deviate from 1 suggests that the symmetric treatment used by the traditional RIP is missing an important part of the picture.  We generalize the RIP to an asymmetric form and derive the sharpest recovery conditions implied by the RIP.

\begin{figure}[t]
 \begin{center}
 \includegraphics[bb= 70 215 546 589, width=3.00 in,height=2.00 in]{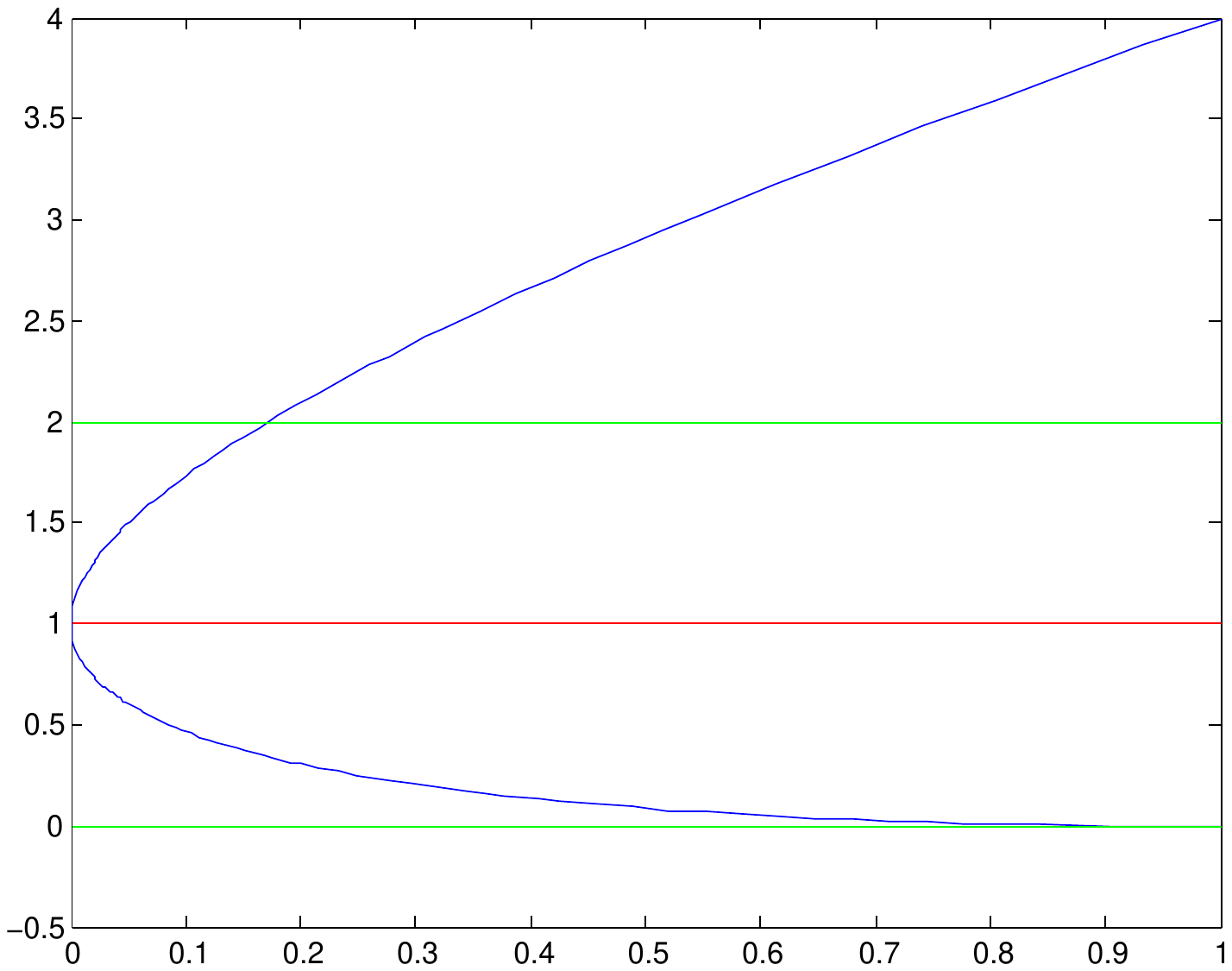}

$\r=k/n$

 \caption{Expected values of the largest and smallest eigenvalues of a Wishart matrix $\wnk$ with $\r=\kn$.  Note the asymmetry with respect to 1.\label{fig:2deigenvalues}}
 \end{center}
 \end{figure}

\begin{figure}[t]
 \begin{center}
 \includegraphics[bb= 70 215 546 589, width=4.30 in,height=2.00 in]{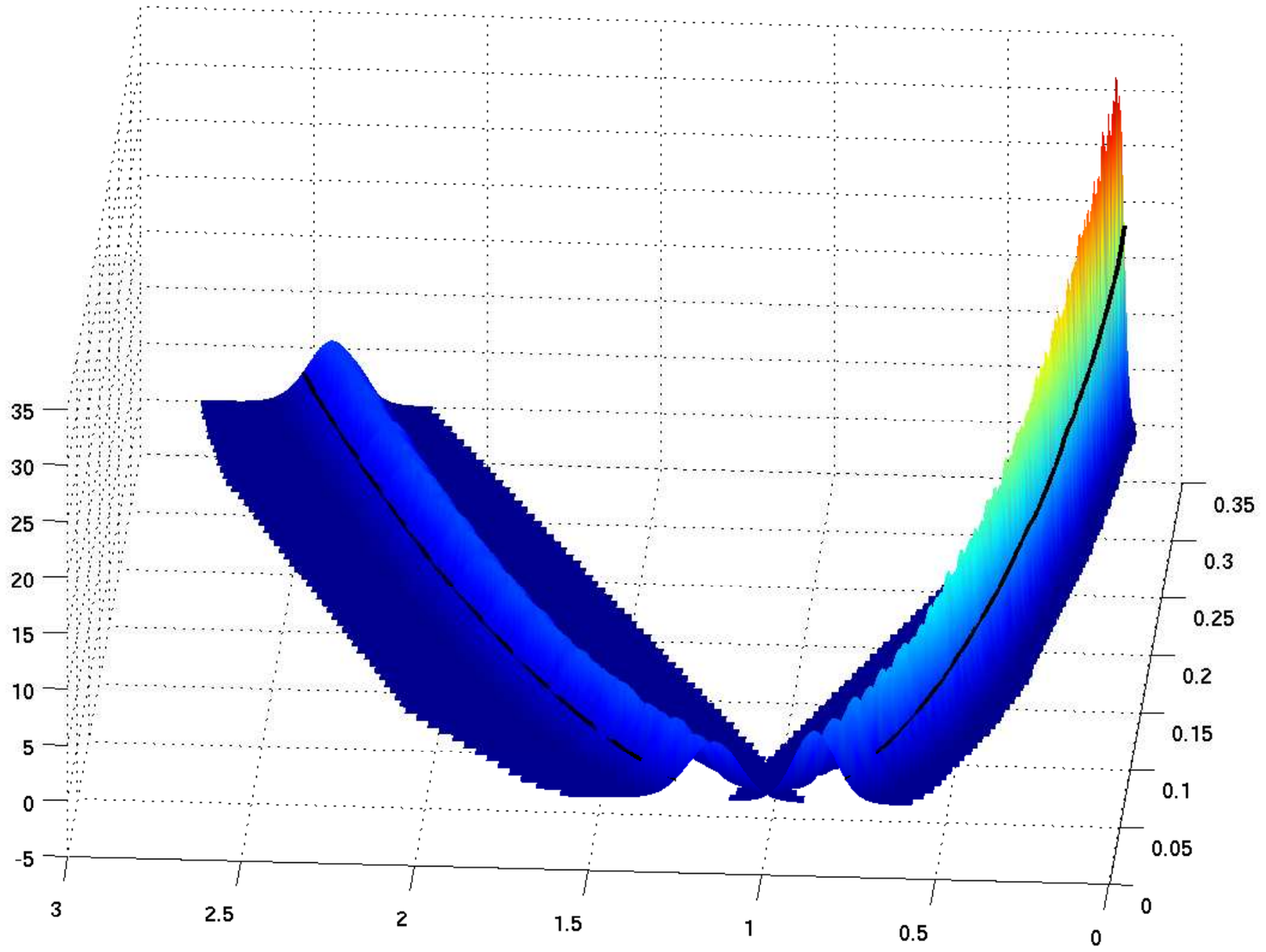}

\vspace{-0.7in}
\hspace{5.1in}
$\rho=\frac{k}{n}$

\vspace{0.7in}

 \caption{{\bf Empirical Distributions of the Largest and Smallest Eigenvalues of a Wishart Matrix.}
 A collection of frequency histograms of $\Lmax$ and $\Lmin$: x-axis -- size of the eigenvalue; y-axis -- number of occurrences; z-axis -- ratio $\r=\kn$ of the Wishart parameters.  Overlays: curves depicting the expected values $(1\pm\sqrt{\r})^2$ of $\Lmax$ and $\Lmin$.  Here $n=200$.  At this value of $n$ it is evident that $\Lmax$ and $\Lmin$ lie near, but not on curves.  For larger $n$, the concentration would be tighter. \label{fig:pdf}}
 \end{center}
 \end{figure}

\begin{definition}[Asymmetric Restricted Isometry Property]\label{def:LU}
For a matrix $A$ of size $n\times N$, the \emph{asymmetric RIP constants} 
$L(k,n,N; A)$ and $U(k,n,N; A)$ are defined as: 
\begin{align}
L(k,n,N;A)&:=\min_{c\ge0}\,c\,\,\,\mbox{subject to}\,\,\,
 (1-c)\|x\|_2^2 \le \|Ax\|_2^2, \,\,\mbox{ for all }\ x\in\chi^N(k);\label{eq:Lrip} \\
U(k,n,N;A)&:=\min_{c\ge0} \,c\,\,\,\mbox{subject to}\,\,\,
(1+c) \|x\|_2^2 \ge \left\|Ax\right\|_2^2, \,\,\mbox{ for all }\  x\in\chi^N(k).\label{eq:Urip}
\end{align}
\end{definition}
 (A similar change in the definition of the RIP constants was used independently 
by Foucart and Lai in \cite{FoucartLai08}, motivated by different concerns.)

\begin{remark}
Although both the smallest and largest
  singular values of $\al^T\al$ affect the stability of the
  reconstruction algorithms, the smaller eigenvalue is dominant for compressed sensing
  in that it allows distinguishing between sparse
  vectors from their measurement by $A$.  In fact, it is often
  incorrectly stated that $R(2k,n,N)<1$ is a necessary condition to ensure that there are no two $k$-sparse vectors, say $x$ and $x'$, with the same measurements $Ax=Ax'$; the actual necessary condition is $L(2k,n,N)<1$.
  \end{remark}

We see from \eqref{eq:Lrip} and \eqref{eq:Urip} 
that $(1-L(k,n,N))=\min_K\lmin(G_K)$ and $(1+U(k,n,N))=\max_K\lmax(G_K)$ with $G_K=\ak^T\ak$.  
A standard large deviation analysis of bounds on the probability density
functions of $\Lmax$ and $\Lmin$ allows us to establish upper bounds  
of $L(k,n,N)$ and $U(k,n,N)$ which are exponentially unlikely to be exceeded.

\begin{definition}[Asymptotic RIP Bounds]\label{def:LUbounds}
 Let $A$ be a matrix of size $n\times N$  drawn  from  the Gaussian ensemble
 and consider the proportional-growth asymptotic ($\nN\rightarrow\d$ and $\kn\rightarrow\r$ as $n\rightarrow\infty$).  Let $H(p):=p\log(1/p)+(1-p)\log(1/(1-p))$ denote the usual Shannon 
Entropy with base $e$ logarithms, and let
\begin{eqnarray}
\psi_{min}(\lambda,\rho) & := & H(\rho)+\frac{1}{2}\left[ 
(1-\rho)\log\lambda +1-\rho+\rho\log\rho-\lambda \right],\label{eq:psimin} \\
\psi_{max}(\lambda,\rho) & := & \frac{1}{2}\left[
(1+\rho)\log\lambda +1+\rho-\rho\log\rho-\lambda \right].\label{eq:psimax} 
\end{eqnarray}
Define $\lmin(\d,\r)$ and 
$\lmax(\d,\r)$ as the solution to \eqref{eq:lmin} 
and \eqref{eq:lmax}, respectively:
\begin{equation}
\delta\psi_{min}(\lambda^{min}(\d,\r),\rho)+H(\rho\delta)=0 
\quad\mbox{ for }\quad \lambda^{min}(\d,\r)\le 1-\r\label{eq:lmin}
\end{equation}
\begin{equation}
\delta\psi_{max}(\lambda^{max}(\d,\r),\rho)+H(\rho\delta)=0
\quad\mbox{ for }\quad \lambda^{max}(\d,\r)\ge 1+\r. \label{eq:lmax}
\end{equation}
Define $\mathcal{L}(\d,\r)$ and $\mathcal{U}(\d,\r)$ as 
\begin{equation}
\mathcal{L}(\d,\r):= 1-\lmin(\delta,\rho) \quad \hbox{and}\quad 
\mathcal{U}(\d,\r):= \min_{\nu\in[\r,1]}\lmax(\d,\nu)-1. \label{eq:LUdeltarho}
\end{equation}
\end{definition} 

To facilitate ease of calculating $\mathcal{L}(\d,\r)$ and $\mathcal{U}(\d,\r)$, web forms for their calculation
are available at \texttt{ecos.maths.ed.ac.uk}.
%, and in the spirit of reproducable research Matlab software to generate all plots is also available.

In the proportional growth asymptotic, the probability that $\mathcal{L}(\d,\r)$ and $\mathcal{U}(\d,\r)$ bound the random variables $L(k,n,N)$ and $U(k,n,N)$, respectively, tends to 1 as $n\rightarrow\infty$.  In statistical terminology, the coverage probability of the upper confidence bounds $\mathcal{L}(\d,\r)$ and $\mathcal{U}(\d,\r)$ tends to one as $n\rightarrow\infty$.  In fact, all probabilities presented in this manuscript converge to their limit ``exponentially in $n$''; that is, the probability for finite $n$ approaches its limit as $n$ grows with discrepancy bounded by a constant multiple of $e^{-n\beta}$ for some fixed $\beta>0$. 

\begin{theorem}[Validity of RIP Bounds]\label{thm:LUbounds}
Fix $\e>0$. Under the proportional-growth asymptotic, Definition \ref{def:pga}, sample each $n\times N$ matrix $A$ from the Gaussian ensemble.  Then
\begin{equation*}
\hbox{Prob}\left(L(k,n,N;A)<\mathcal{L}\left(\d,\r\right)+\e\right)\rightarrow 1\quad \hbox{and}\quad \hbox{Prob}\left(U(k,n,N;A)<\mathcal{U}\left(\d,\r\right)+\e\right)\rightarrow 1
\end{equation*}
exponentially in $n$.
\end{theorem}

\begin{figure}[t]
 \begin{center}
 \includegraphics[bb= 70 215 546 589, width=2.70 in,height=2.00 in]{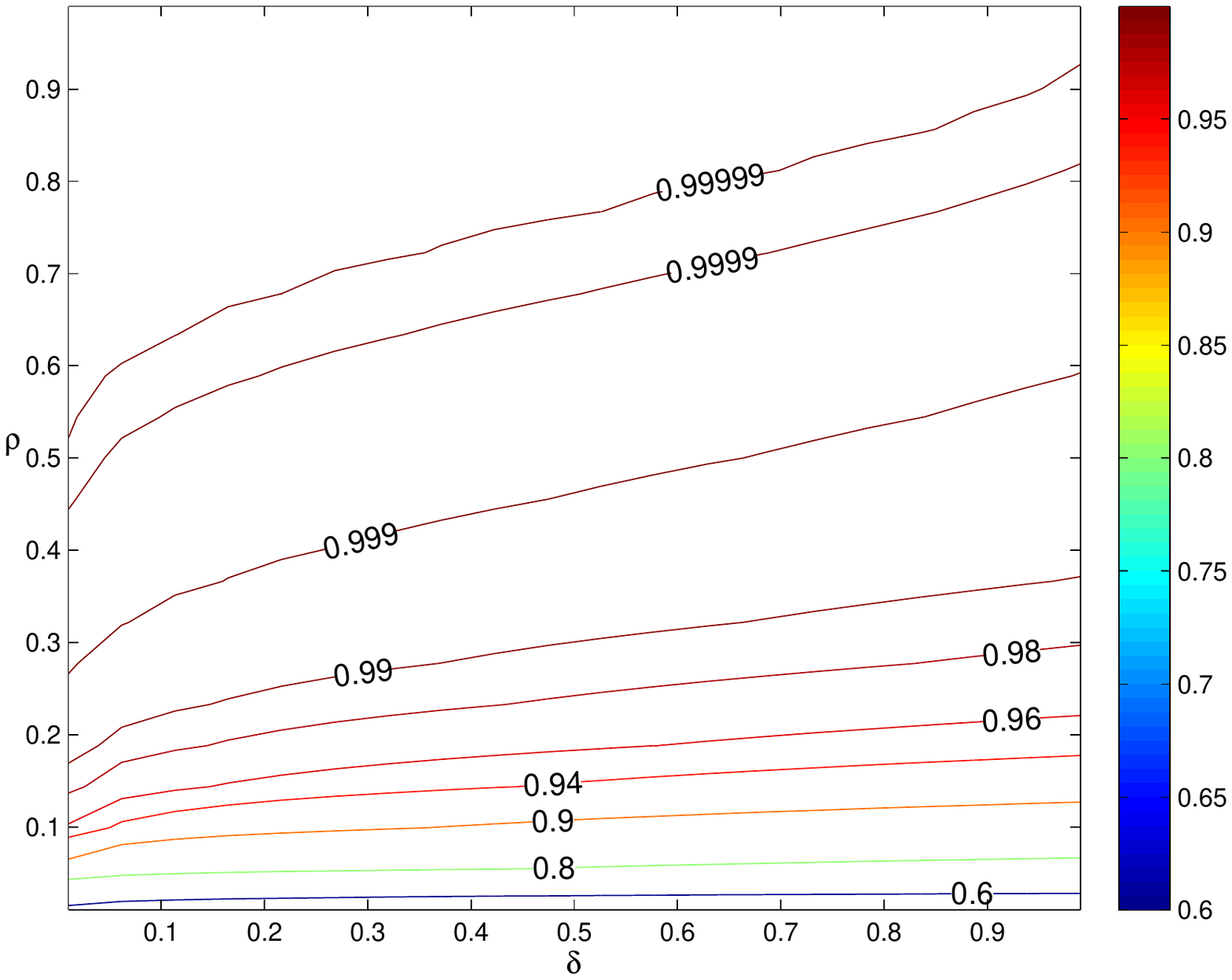}
\includegraphics[bb= 70 215 546 589, width=2.70 in,height=2.00 in]{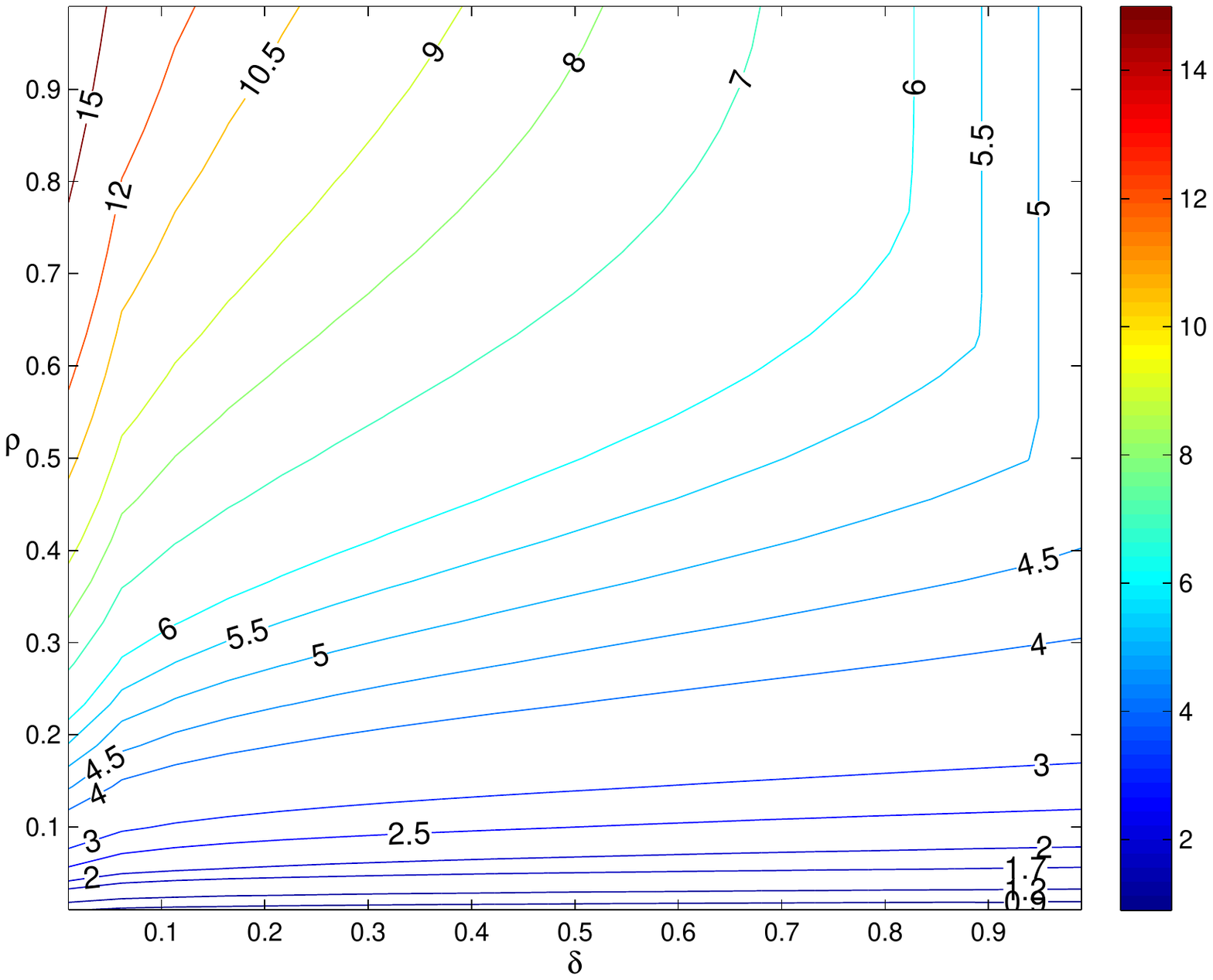}
\end{center}
 \caption{{\bf The RIP bounds of Eq. \eqref{eq:LUdeltarho}.}  Level sets 
 of $\mathcal{L}(\d,\r)$ (left panel) and $\mathcal{U}(\d,\r)$ (right panel) over the phase space $(\d,\r)\in[0,1]^2$.  For large matrices from the Gaussian ensemble, it is overwhelmingly unlikely that the RIP constants $L(k,n,N;A)$ and $U(k,n,N;A)$ will be greater than these values.\label{fig:LU}}
 \end{figure}

\begin{figure}[t]
 \begin{center}
 \includegraphics[bb= 70 215 546 589, width=2.70 in,height=2.00 in]{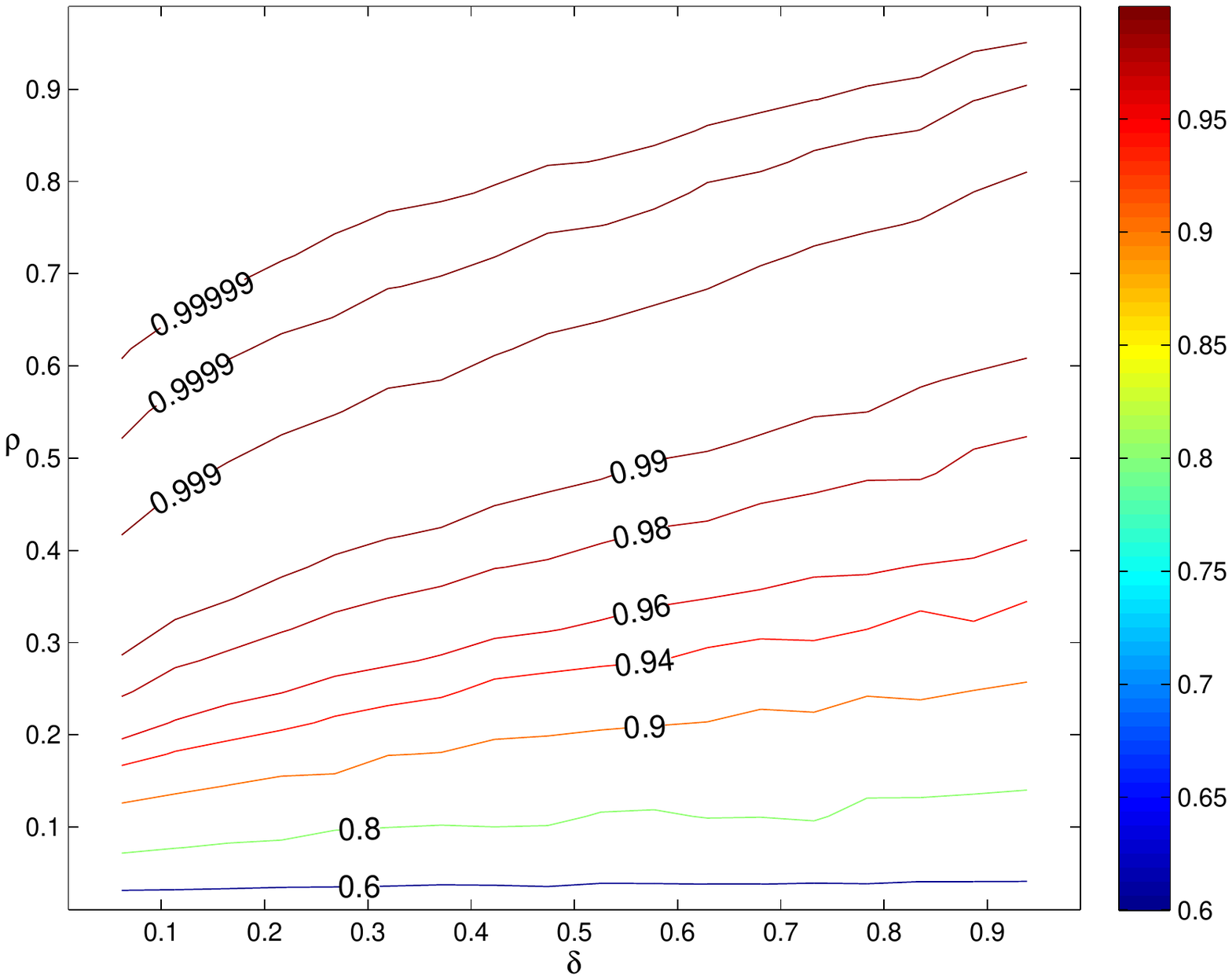}
\includegraphics[bb= 70 215 546 589, width=2.70 in,height=2.00 in]{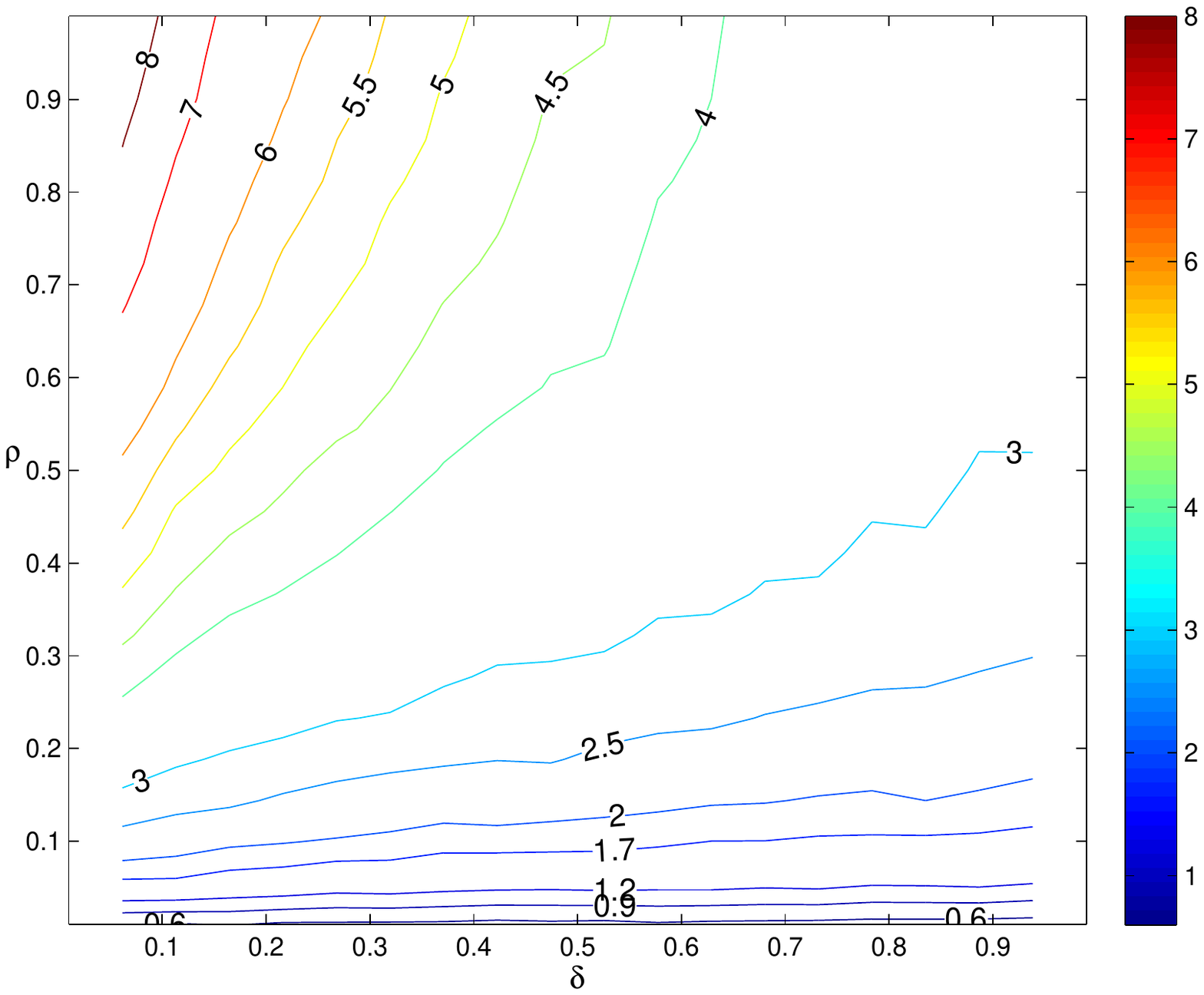}
\end{center}
 \caption{{\bf Empirically observed lower estimates of RIP bounds of RIP constants.}  Although there is no computationally tractable method for calculating the RIP constants of a matrix, there are efficient algorithms which perform local searches for extremal eigenvalues of submatrices; allowing for observable {\em lower} bounds on the RIP constants.  Algorithms for lower bounding $L(k,n,N)$, \cite{Dossal}, and $U(k,n,N)$, \cite{Richtarik}, were applied to dozens of $A$ drawn Gaussian ${\cal N}(0,n^{-1})$ with $n=400$ and $N$ increasing from 420 to 8000.  Level sets of the observed $L(k,n,N;A)$ (left panel) and $U(k,n,N;A)$ (right panel). \label{fig:lu_empirical}}
\end{figure}

\begin{remark}
Extensive empirical estimates of $L(k,n,N)$ and $U(k,n,N)$ show that the bounds
$\mathcal{L}(\d,\r)$ and $\mathcal{U}(\d,\r)$
 are rather sharp; in fact, they are no more than twice the actual upper bounds on $L(k,n,N)$ and $U(k,n,N)$, see Figure \ref{fig:lu_empirical} and Table \ref{tab:emp}, and are much closer for the region applicable for CS decoders, $\rho\ll 1$.  The empirically observed lower bounds on $L(k,n,N)$ and $U(k,n,N)$ are calculated through the following process.  The number of rows, $n$, is fixed at one of the values in Table \ref{tab:emp}.  For each $n$, 47 values of $N$ are selected so that $n/N$ ranges from $1/20$ to $20/21$.  For each $(n,N)$ a matrix $A$ of size $n\times N$ is drawn from ${\cal N}(0,n^{-1})$ and either the algorithm from \cite{Dossal} or \cite{Richtarik} is applied to determine support sets of size $k=1,2,\ldots,n-1$ which are candidates for the support sets that maximize $L(k,n,N;A)$ or $U(k,n,N;A)$.  The largest or smallest eigenvalue of each resulting $n\times k$ submatrix is calculated and recorded.  The above process is repeated for some number of matrices, see the caption of Table \ref{tab:emp}, and the maximum value recorded.  The empirical calculation of RIP constants are {\em lower bounds} on the true RIP constants as the support sets calculated by \cite{Dossal} and \cite{Richtarik} may not be the support sets which maximize the RIP constants.
\end{remark}

\begin{table}[h!]
\begin{center}
\begin{tabular}{|c|c|c|}
\hline
$n$ & $\max \frac{\mathcal{L}(\d,\r)}{L(k,n,N)}$ & $\max \frac{\mathcal{U}(\d,\r)}{U(k,n,N)}$ \\
\hline\hline
200 & 1.22 & 1.83 \\
\hline
400 & 1.32 & 1.81 \\
\hline
\end{tabular}
\end{center}
\caption{The maximum ratio of the RIP bounds in Theorem \ref{thm:LUbounds} to empirically observed values.  For each of the ratios $n/N$ tested, multiple matrices were drawn and empirical low bounds on their RIP constants calculated.  For $n=200$ between 9 and 175 matrices were drawn for each $n/N$, and for $n=400$ between 7 and 489 matrices were drawn for each $n/N$.  Our bounds are numerically found to be within a multiple of 1.83 of empirically observed lower bounds.
\label{tab:emp}}
\end{table}

\subsection{Proof of Theorem \ref{thm:LUbounds}}\label{sec:Proofs}

In order to prove Theorem~\ref{thm:LUbounds},
this section employs a type of large deviation technique often  encountered in CS and applicable in fact, to many areas faced with combinatorial challenges.

We first establish some useful lemmas concerning the extreme eigenvalues of Wishart matrices.  The matrix $A$ generates ${N\choose k}$ different Wishart matrices $G_k=\ak^T\ak$.  Exponential bounds on the tail probabilities of the largest and smallest eigenvalues of such Wishart matrices can be combined with exponential bounds on ${N\choose k}$ to control the chance of large deviations using the union bound.  This large deviation analysis technique is characteristic of proofs in compressed sensing.  By using the exact probability density functions on the tail behavior of the extreme eigenvalues of Wishart matrices the overestimation of the union bound is dramatically reduced.   We focus on the slightly more technical results for the bound on the most extreme of the largest eigenvalues, $\mathcal{U}(\d,\r)$, and prove these statements in full detail.  Corresponding results for $\mathcal{L}(\d,\r)$ are stated with their similar proofs omitted.   

The probability density function, $\pdfmax$, for the largest eigenvalue of the $k\times k$ Wishart matrix $\al^T\al$ was determined by Edelman in \cite{EdelmanEigenvalues88}.  For our analysis, a simplified upper bound suffices.

\begin{lemma}[Lemma 4.2, pp. 550 \cite{EdelmanEigenvalues88}]\label{lem:EdelmanMax}
Let $\al$ be a matrix of size $n\times k$ whose entries are drawn i.i.d from ${\cal N}(0,n^{-1})$.  Let $\pdfmax$ denote the probability density function for the largest eigenvalue of the Wishart matrix $\al^T\al$ of size $k\times k$.  Then $\pdfmax$ satisfies:
\begin{equation}\label{eq:pdf_max}
\pdfmax\le \left[(2\pi)^{1/2}(n\lambda)^{-3/2} 
\left(\frac{n\lambda}{2}\right)^{(n+k)/2}\frac{1}{\Gamma(\frac{k}{2})\Gamma(\frac{n}{2})}\right]\cdot e^{-n\lambda/2}=:g_{max}(k,n;\l).
\end{equation}
\end{lemma}

For our purposes, it is sufficient to have a precise characterization of $g_{max}(k,n;\l)$'s exponential (with respect to $n$) behavior.

\begin{lemma}\label{lem:ASYMmax}
Let $k/n=\r\in(0,1)$ and define
\begin{equation*}
\psi_{max}(\lambda,\rho):=\frac{1}{2}\left[
(1+\rho)\log\lambda +1+\rho-\rho\log\rho-\lambda
\right].
\end{equation*}
Then
\begin{equation}
\pdfmax\le p_{max}(n,\lambda)\exp(n\cdot\psi_{max}(\lambda,\rho)) 
\end{equation}
where $p_{max}(n,\l)$ is a polynomial in $n,\l$.
\end{lemma}

\begin{proof} Let $g_{max}(k,n;\l)$ be as defined in \eqref{eq:pdf_max} and let $\rn=k/n$.  To extract the exponential behavior of $g_{max}(k,n;\l)$ we write
$\frac{1}{n}\log(g_{max}(k,n;\l))=\Phi_1(k,n;\l)+\Phi_2(k,n;\l)+\Phi_3(k,n;\l)$ where
\begin{align*}
\Phi_1(k,n;\l) &= \frac1{2n}\log\left(2\pi\right) -\frac3{2n}\log\left(n\l\right) \\
\Phi_2(k,n;\l) &=  \frac12\left[(1+\rn)\log\left(\frac{\l n}{2}\right)-\l\right]  \\
\Phi_3(k,n;\l) &= -\frac1n\log\left(\Gamma\left(\frac{k}{2}\right)\Gamma\left(\frac{n}{2}\right)\right). 
\end{align*}
Clearly, $\lim_{n\rightarrow\infty}\Phi_1(k,n;\l)=0$ and can be subsumed as part of $p_{max}(n,\l)$.  To simplify $\Phi_3$, we apply the second of Binet's log gamma formulas \cite[Sec. 12.32]{WhittakerWatson}, namely $\log(\Gamma(z))=(z-1/2)\log z - z + \log\sqrt{2\pi} + I$ where $I$ is a convergent, improper integral.  With $c(n,\r)$ representing the constant and integral from Binet's formula we then have
\begin{equation*}
\Phi_2(k,n;\l)+\Phi_3(k,n;\l) = \frac12\left[(1+\rn)\log\l-\left(\rn-\frac1n\right)\log\rn+\frac2n\log\frac{n}2+\rn+1-\l+\frac1nc(n,\rn)\right]. 
\end{equation*}
As $\lim_{n\rightarrow \infty}n^{-1}c(n,\r_n)=0$ it can be absorbed into $p_{max}(n,\lambda)$ and we have
\begin{equation*}
\psi_{max}(\l,\r):=\lim_{n\rightarrow\infty}\frac1n\log \left[g_{max}(k,n;\l)\right]=\frac12\left[(1+\r)\log\l-\r\log\r+\r+1-\l\right]
\end{equation*} and the conclusion follows.
\end{proof}

To bound $U(k,n,N)$, we must simultaneously account for all ${N\choose k}$ Wishart matrices $\al^T\al$ derived from $A$.  Using a union bound this amounts to studying the exponential behavior of ${N\choose k}g_{max}(k,n;\l)$.  In the proportional-growth asymptotic this can be determined by characterizing $\lim_{N\rightarrow\infty}N^{-1}\log\left[{N\choose k}g_{max}(k,n;\l)\right]$, which from Lemma \ref{lem:ASYMmax} is given by 
\begin{align}
\lim_{N\rightarrow\infty}\frac1N\log\left[{N\choose k}g_{max}(k,n;\l)\right] &= \lim_{N\rightarrow\infty}\frac1N\log\left[{N\choose k}\right]
+ \lim_{N\rightarrow\infty}\frac1N\log\left[g_{max}(n,k;\l)\right] \nonumber \\
&= H\left(\frac{k}{N}\right) + \d\lim_{n\rightarrow\infty}\frac1n \log\left[g_{max}(n,k;\l)\right] \nonumber\\
&= H(\r\d) + \d\psi_{max}(\l,\r)=:\d\psi_\mathcal{U}(\d,\r;\l).\label{eq:LDmax}
\end{align}
Recall that $H(p):=p\log(1/p)+(1-p)\log(1/(1-p))$ is the usual Shannon Entropy with base $e$ logarithms. 

Equipped with Lemma \ref{lem:ASYMmax} and \eqref{eq:LDmax}, Proposition \ref{prop:U} establishes $\lmax(\delta,\rho)-1$ as an upper bound on $U(k,n,N)$ in the proportional-growth asymptotic.

\begin{proposition}\label{prop:U} 
Let $\d,\r\in(0,1)$, and $A$ be a matrix of size $n\times N$ whose entries are drawn i.i.d.\ from ${\cal N}(0,n^{-1})$.  Define $\tilde{\mathcal{U}}(\d,\r):= \lmax(\delta,\rho)-1$ where $\lmax(\d,\r)$ is the solution to \eqref{eq:lmax}. Then for any $\e>0$, in the proportional-growth asymptotic 
\[Prob\left( U(k,n,N)>\tilde{\mathcal{U}}(\delta,\rho)+\e\right)\rightarrow 0\] exponentially in $n$.
\end{proposition}

\begin{proof}  Throughout this proof $\d$ and $\r$ are fixed, and we focus our attention on $\l$, often abbreviating $\psi_\mathcal{U}(\d,\r;\l)$  in \eqref{eq:LDmax} as $\psi_U(\l)$.  We first verify that \eqref{eq:lmax} has a unique solution.  Since
\[ 
\frac{d}{d\l}\psi_U(\l)=\frac12\left(\frac{1+\r}{\l}-1\right),
\]
$\psi_U(\l)$ is strictly decreasing on $[1+\r,\infty)$ and is strictly concave.
Combined with
\begin{align*}
\psi_U(1+\r) &=\d^{-1}H(\r\d)+\frac12\left[(1+\r)\log(1+\r) +\r\log\frac{1}{\r}\right]>0
\end{align*}
and $\lim_{\l\rightarrow\infty}\psi_U(\l)=-\infty$, there is a unique solution to \eqref{eq:lmax}, namely $\lmax(\d,\r)$.

Select $\e>0$ and let $(k,n,N)$ be such that $\nN= \dn$, $\kn=\rn$.  First, we write the probability statement in terms of $\lmax(\dn,\rn)$:
\begin{align}
Prob\left[U(k,n,N)>\tilde{\mathcal{U}}(\dn,\rn)+\e)\right] &= Prob\left[U(k,n,N)>\lmax(\dn,\rn)-1+\e)\right] \nonumber \\
&= Prob\left[1+U(k,n,N)>\lmax(\dn,\rn)+\e)\right] \nonumber \\
&= {N\choose k}\int_{\lmax(\dn,\rn)+\e}^\infty \pdfmax d\l \nonumber \\
&\le {N\choose k}\int_{\lmax(\dn,\rn)+\e}^\infty g_{max}(k,n;\l) d\l. \label{eq:intgmax}
\end{align}

To bound the integral in \eqref{eq:intgmax} in terms of $g_{max}(\d,\r;\lmax(\dn,\rn))$ we write $\gmax$ in terms of $n$, $\rn$, and $\l$ as $\gmax=\varphi(n,\rn)\l^{-\frac32}\l^{\frac{n}2(1+\rn)}e^{-\frac{n}2\l}$ where 
\[ \varphi(n,\rn)=(2\pi)^\frac12 n^{-\frac32} \left(\frac{n}{2}\right)^{\frac{n}2(1+\rn)} \frac1{\G\left(\frac{n}2\rn\right)\G\left(\frac{n}2\right)}.\]
Since $\lmax(\dn,\rn)>1+\rn$, the quantity $\l^{\frac{n}2(1+\rn)}e^{-\frac{n}2\l}$ is strictly decreasing in $\l$ on $[\lmax(\d,\rn),\infty)$.  Therefore we have
\begin{align}
\int_{\lmax(\dn,\rn)+\e}^\infty g_{max}(k,n;\l) d\l &\le \varphi(n,\rn)\left(\lmax(\dn,\rn)+\e\right)^{\frac{n}2(1+\rn)}e^{-\frac{n}2\left(\lmax(\dn,\rn)+\e\right)}\int_{\lmax(\dn,\rn)+\e}^\infty \l^{-\frac32} d\l \nonumber \\
&= \left(\lmax(\dn,\rn)+\e\right)^{\frac32}g_{max}\left(k,n;\lmax(\dn,\rn)+\e\right)\int_{\lmax(\dn,\rn)+\e}^\infty \l^{-\frac32} d\l \nonumber \\
&= 2\left(\lmax(\dn,\rn)+\e\right)g_{max}\left(k,n;\lmax(\dn,\rn)+\e\right). \label{eq:intbound}
\end{align}
Therefore, combining \eqref{eq:intgmax} and \eqref{eq:intbound} we obtain
\begin{align}
Prob\left[U(k,n,N)>\tilde{\mathcal{U}}(\dn,\rn)+\e)\right] &\le 2\left(\lmax(\dn,\rn)+\e\right){N\choose k}g_{max}\left(k,n;\lmax(\dn,\rn)+\e\right) \nonumber \\
&\le p_{max}\left(n,\lmax(\dn,\rn)\right)\exp\left[n\cdot \psi_U\left(\lmax(\dn,\rn)+\e\right)\right] \nonumber \\ 
& \le p_{max}\left(n,\lmax(\dn,\rn)\right)\exp\left[n\e\cdot 
\frac{d}{d\l}\psi_U(\l)|_{\l=\left(\lmax(\dn,\rn)\right)}\right], \label{eq:Ulastboud}
\end{align}
with the last inequality following from the strict concavity of $\psi_U(\l)$.
Since $\frac{d}{d\l}\psi_U\left(\lmax(\d,\r)\right)<0$ is strictly bounded away from zero and 
$\lim_{n\rightarrow \infty}\lmax(\dn,\rn)=\lmax(\d,\r)$, we arrive at, for any $\e>0$
\[
\lim_{n\rightarrow\infty}Prob\left[U(k,n,N)>\tilde{\mathcal{U}}(\d,\r)+\e)\right]\rightarrow 0.
\]

\end{proof}

The term $H(\r\d)$ in \eqref{eq:LDmax}, from the union bound over all $N \choose k$ matrices $\al^T\al$, results in an overly pessimistic bound in the vicinity of $\r\d= 1/2$.  As we are seeking the least upper bound on $U(k,n,N)$ we note that any upper bound for $U(j,n,N)$ for $j>k$ is also an upper bound for $U(k,n,N)$, and replace the bound $\tilde{\mathcal{U}}(\d,\r)$ with the minimum of $\tilde{\mathcal{U}}(\d,\nu)$ for $\nu\in[\r,1]$.

\begin{proposition}\label{prop:U2} 
Let $\d,\r\in(0,1)$, and define $\mathcal{U}(\d,\r):=\min_{\nu\in [\r,1]}\tilde{\mathcal{U}}(\d,\nu)$ with $\tilde{\mathcal{U}}(\d,\nu)$ defined as in Proposition \ref{prop:U}.  For any $\e>0$, in the 
proportional-growth asymptotic 
\[Prob\left( U(k,n,N)>\mathcal{U}(\delta,\rho)+\e\right)\rightarrow 0\] exponentially in $n$.
\end{proposition}

\begin{proof}
By the definition of $\chi^N(k)$ in Definition \ref{def:rip}, $U(j,n,N)\ge U(k,n,N)$ for $j=k+1,k+2,\ldots,n$; combined with Proposition \ref{prop:U}
for $\frac{j}{n}\rightarrow\nu$ as $n\rightarrow\infty$
\[Prob\left( U(j,n,N)>\tilde{\mathcal{U}}(\delta,\nu)+\e\right)\rightarrow 0\]
exponentially in $n$, and taking a minimum over the compact set $\nu\in [\r,1]$ we arrive at the desired result.
\end{proof}

A similar approach leads to corresponding results for $\mathcal{L}(\d,\r)$.  Edelman also determined the probability density function, $\pdfmin$, for the smallest eigenvalue of the $k\times k$ Wishart matrix $\al^T\al$ \cite{EdelmanEigenvalues88}.  Here again, a simplified upper bound suffices:

\begin{lemma}[Prop. 5.2, pp. 553 \cite{EdelmanEigenvalues88}]\label{lem:EdelmanMin}
Let $\al$ be a matrix of size $n\times k$ whose entries are drawn i.i.d.\ from $\mathcal{N}(0,n^{-1})$.  Let $\pdfmin$ denote the probability density function for the smallest eigenvalue of the Wishart matrix $\al^T\al$ of size $k\times k$.  Then $\pdfmin$ satisfies:
\begin{equation}\label{eq:pdf_min}
\pdfmin\le 
\left(\frac{\pi}{2n\lambda}\right)^{1/2}\;\cdot\;
e^{-n\lambda/2}\left(\frac{n\lambda}{2}\right)^{(n-k)/2}
\;\cdot\;
\left[    \frac{\Gamma(\frac{n+1}{2})}{\Gamma(\frac{k}{2})\Gamma(\frac{n-k+1}{2})\Gamma(\frac{n-k+2}{2})}\right]=:g_{min}(k,n;\l).
\end{equation}
\end{lemma}

With Lemma \ref{lem:EdelmanMin}, we establish a bound on the asymptotic behavior of the distribution of the smallest eigenvalue of Wishart matrix of size $k\times k$. 

\begin{lemma}\label{lem:ASYMmin}
Let $k/n=\rho\in(0,1)$ and define
\begin{equation*}
\psi_{min}(\lambda,\rho):=H(\rho)+\frac{1}{2}\left[(1-\rho)\log\lambda +1-\rho+\rho\log\rho-\lambda\right].
\end{equation*}
Then
\begin{equation} \label{eq:LargeDeviationMin}
\pdfmin\le p_{min}(n,\lambda)\exp(n\cdot\psi_{min}(\lambda,\rho)) 
\end{equation}
where $p_{min}(n,\l)$ is a polynomial in $n,\l$.
\end{lemma}

With Lemma \ref{lem:ASYMmin}, the large deviation analysis yields
\begin{equation}
\lim_{N\rightarrow\infty}\frac1N\log\left[{N\choose k}g_{min}(k,n;\l)\right] = H(\r\d) + \d\psi_{min}(\l,\r).\label{eq:LDmin}
\end{equation}
Similar to the proof of Proposition \ref{prop:U}, Lemma \ref{lem:ASYMmin} and \eqref{eq:LDmin} are used to establish $\mathcal{L}(\d,\r)$ as an upper bound on $L(k,n,N)$ in the proportional-growth asymptotic.  

\begin{proposition}\label{prop:L} 
Let $\d,\r\in(0,1]$, and $A$ be a matrix of size $n\times N$ whose entries are drawn i.i.d.\ from ${\cal N}(0,n^{-1})$. Define $\mathcal{L}(\d,\r):= 1-\lmin(\delta,\rho)$ where $\lmin(\d,\r)$ is the solution to \eqref{eq:lmin}.  Then for any $\e>0$, in the proportional-growth asymptotic 
\[ Prob\left( L(k,n,N)>\mathcal{L}(\delta,\rho)+\e\right)\rightarrow 0\] exponentially in $n$.
\end{proposition}

The bound $\mathcal{L}(\d,\r)$ is strictly increasing in $\r$ for any $\d\in(0,1)$, and as a consequence no tighter bound can be achieved by minimizing over matrices of larger size as was done in Proposition \ref{prop:U2}.

%%%%%%%%%%%%%%%%%%%%%%%% Phase Transitions %%%%%%%%%%%%%%%%%%%%%%%

\section{RIP Undersampling Theorems}\label{sec:Algorithms}

The high level of interest in compressed sensing is due to the introduction of computationally efficient and stable algorithms which provably solve the seemingly intractable \eqref{eq:l0prime} even for $k$ proportional to $n$.  New compressed sensing decoders are being introduced regularly; broadly speaking, they fall into one of two categories: greedy algorithms
and regularizations.  Greedy algorithms are iterative, with each step
selecting a locally optimal subset of entries in $x$ which are adjusted to improve the
desired error metric.  Examples of greedy algorithms include
Orthogonal Matching Pursuit (OMP) \cite{Greed}, Regularized OMP (ROMP)
\cite{NeVe06_UUP}, Stagewise OMP (StOMP) \cite{DTDS08}, Compressive 
Sampling MP (CoSaMP) \cite{NeTr09_cosamp}, Subspace Pursuit (SP) \cite{SubspacePursuit}, and Iterated Hard Thresholding (IHT) \cite{BlDa08_iht}.  Regularization 
formulations for sparse approximation
began with the relaxation of \eqref{eq:l0prime} to the now ubiquitous (convex) $\ell^1$-minimization \cite{CDS98}, \eqref{eq:l1}, and has since been extended to non-convex 
$\ell^q$-minimization for $q\in(0,1)$, \cite{Gribonval,FoucartLai08,ChartrandRIP,ChartrandExact,Yilmaz}.
Although general-purpose convex optimization solvers may be employed 
to solve $\ell^1$-minimization  \eqref{eq:l1}, highly-efficient  software  has been recently designed specifically for $\ell^1$-minimization in the context of compressed sensing, see \cite{CDS98, GPSR, BeFr08, Bregman}.  Non-convex formulations have sometimes been able to offer substantial improvements, but at the cost of limited guarantees that the global minima can be found efficiently, so it remains unclear how practical they really are.

As stated at the end of the introduction, one of the central aims of this article is to advocate a unifying framework for the
comparison of results in compressed sensing. Currently there is no general agreement in the compressed sensing community on such a framework, making it difficult to compare results obtained by different methods of analysis or to identify when new results are improvements over existing ones.  Donoho has put forth the phase transition framework borrowed
from the statistical mechanics literature  and used
successfully in a similar context by the combinatorial optimization
community, see \cite{HartmannRieger,HartmannWeight}.  This framework has been successfully employed in compressed sensing by Donoho et al, 
\cite{DoSt06_breakdown,DoTa05_signal,DoTs06_fast}.

Fortunately, every compressed sensing algorithm that has an optimal recovery order of $n$ proportional to $k$ can be cast in the phase transition framework of Donoho et al., parametrized by two inherent problem size 
parameters\footnote{For some algorithms, such as $\ell^1$-regularization, 
these two parameters fully 
characterize the behavior of the algorithm for a particular matrix ensemble, 
whereas for other algorithms, such as OMP, the distribution of the 
nonzero coefficients also influences the behavior of the method.}:
\begin{itemize}  
\item the {\em undersampling rate} of measuring $x$ through $n$ inner products
with the rows of $A$, as compared to directly sampling each
element of $x\in\RR^N$: 
$$\d_n=n/N\in(0,1)$$
\item  the {\em oversampling rate} of making $n$
measurements as opposed to the optimal {\em oracle} rate of making
$k$ measurements when the oracle knows the support of $x$:
$$\r_n=k/n\in (0,1).$$ 
\end{itemize}
For each value of $\d_n\in (0,1)$ there is a largest 
value of $\rho_n$ which guarantees successful recovery of $x$.

We now formalize the phase transition framework described above. 

\begin{definition}[Strong Equivalence] The event StrongEquiv($A,${\sc alg}) denotes the following property of an $n\times N$ matrix $A$: for \emph{every} $k$-sparse vector $x$, the algorithm ``{\sc alg}'' exactly recovers $x$ from the corresponding measurements $y=Ax$.
\end{definition}

For most compressed sensing algorithms and for a broad class of matrices, under the proportional-growth asymptotic there is a strictly positive
function $\r_S(\delta;${\sc alg}$)>0$ defining a region of the $(\d,\r)$ phase space which ensures successful recovery of every $k$-sparse vector $x\in\chi^N(k)$. This function, $\r_S(\d;${\sc alg}), is called the \emph{Strong phase transition} function \cite{CTDecoding,Do05_signal,Do05_polytope}.  

\begin{definition}[Region of Strong Equivalence]  
Consider the proportional-growth asymptotic with parameters $(\d,\r)\in(0,1)\times(0,1/2)$.  Draw the corresponding $n\times N$ matrices $A$ from the Gaussian ensemble and fix $\e>0$.  Suppose that we are given a function $\r_S(\d$;{\sc alg}) with the property that, whenever $0<\r<(1-\e)\r_S(\d;${\sc alg}), $\hbox{Prob}(\hbox{StrongEquiv}(A, ${\sc alg}$))\rightarrow1$ as $n\rightarrow\infty$.  We say that $\rho_S(\d$;{\sc alg}) bounds a \emph{region of strong equivalence}.
\end{definition}

\begin{remark}
The subscript $S$ emphasizes that the phase transition function $\r_S(\d$;{\sc alg}) 
will define a region of the $(\d,\r)$ phase space which guarantees that the event 
StrongEquiv($A$,{\sc alg}) is satisfied with  probability on the draw of $A$ converging to one exponentially in $n$.  This notation has been established in the literature by Donoho and Tanner \cite{Do05_signal,DoTa05_signal} to distinguish strong equivalence (i.e. that every $k$-sparse vector $x$ is successfully recovered) from weak equivalence (i.e. all but a small fraction of $k$-sparse vectors are successfully recovered).  For example, \cite{Do05_signal,DoTa05_signal} study the event  where  $\lone$-minimization \eqref{eq:l1} exactly recovers $x$ from the corresponding measurements $y=Ax$, except for a fraction $(1-\e)$ of the support sets.
\end{remark}

For the remainder of this section, we translate existing guarantees of $\sl1$ into bounds on the region of strong equivalence in the proportional-growth asymptotic; we denote $\r_S(\d;\ell^1)\equiv
\r_S(\d)$.   A similar presentation of other CS decoders is available in \cite{BCTT09}.   In order to make quantitative statements, the matrix or random matrix ensemble must first be specified, \cite{BCT09_DecayRIP};  we again consider $A$ drawn from the Gaussian ensemble.\footnote{Similar results have been
proven for other random matrix ensembles, but they are even less precise than those for the Gaussian distribution.}  In Section \ref{subsec:l1l0} we demonstrate how to incorporate the RIP bounds from Section~\ref{sec:RIPBounds} into results obtained from an RIP analysis.  In Section~\ref{subsec:comparison} we compare bounds on the region of $\sl1$ proven by three distinct methods of analysis: eigenvalue analysis and the RIP \cite{FoucartLai08},  geometric functional analysis \cite{RV07_gaussian}, and convex polytopes \cite{Do05_signal}.  

\subsection{Region of $\sl1$ implied by the RIP}\label{subsec:l1l0}

In this section, we incorporate the bounds on RIP constants established in Section~\ref{sec:RIPBounds} into a known condition implying $\sl1$ obtained from an RIP analysis.  Following the pioneering work of Cand\`es, Romberg, and Tao \cite{CRTRobust,CTNearOptimal}, many different conditions on the RIP constants have been developed which ensure recovery of every $k$-sparse vector via $\lone$-minimization,  
\cite{CompressiveSampling,CandesRoot2,CRTStable, CTDecoding,RV07_gaussian} to name a few.  The current state of the art RIP conditions for $\ell^1$-minimization were developed by Foucart and Lai \cite{FoucartLai08}.

\begin{theorem}[Foucart and Lai \cite{FoucartLai08}]\label{thm:LF1}
  For any matrix $A$ of
  size $n\times N$ with RIP constants $L(2k,n,N)$ and $U(2k,n,N)$, for
  $2k\le n<N$.  Define  
\begin{equation}\label{eq:mu_knN}
\mu^{FL}(k,n,N):=\frac{1+\sqrt{2}}{4}\left(\frac{1+U(2k,n,N)}{1-L(2k,n,N)} -1 \right).
\end{equation}
If $\mu^{FL}(k,n,N)<1$, then there is $\sl1$.
\end{theorem} 

To translate this result into the phase transition framework for matrices from the Gaussian ensemble, we employ the RIP bounds \eqref{eq:LUdeltarho} to the asymmetric RIP constants $L(2k,n,N)$ and $U(2k,n,N)$. It turns out that naively inserting these bounds into \eqref{eq:mu_knN} yields a bound on $\mu^{FL}(k,n,N)$, see Lemma \ref{lem:mu_epsilon2}, and provides a simple way to obtain a bound on the region of strong equivalence.

\begin{definition}[RIP Region of $\sl1$ ]
Define
\begin{equation}\label{mu_dr}
\mu^{FL}(\d,\r):=\frac{1+\sqrt{2}}{4}\left(\frac{1+\mathcal{U}(\d,2\r)}{1-\mathcal{L}(\d,2\r)} -1 \right)
\end{equation}
and $\rfl$ as the solution to $\mu^{FL}(\d,\r)=1$.
\end{definition}

The function $\rfl$ is displayed as the red curve in Figure \ref{fig:l1}.  

\begin{theorem}\label{thm:LF1_asymptotic}
Fix $\e>0$.  Consider the proportional-growth asymptotic, Definition \ref{def:pga},  with parameters $(\d,\r)\in(0,1)\times(0,1/2)$.  Draw the corresponding $n\times N$ matrices $A$ from the Gaussian ensemble.  If $\r<(1-\e)\rfl$, then $\hbox{Prob}(\hbox{StrongEquiv}(A,\lone))\rightarrow1$ as $n\rightarrow\infty$.

Therefore the function $\rfl$ bounds a region of strong equivalence for $\lone$-minimization.
\end{theorem} 

Theorem \ref{thm:LF1_asymptotic} follows from Theorem \ref{thm:LF1}
and the validity of the probabilistic bounds on the RIP constants, Theorem \ref{thm:LUbounds}.
In particular, Lemma \ref{lem:mu_epsilon2} bounds $\mu^{FL}(k,n,N)$ in terms of 
the asymptotic RIP bounds $\mathcal{L}(\d,2\r)$ and $\mathcal{U}(\d,2\r)$, by the quantity
$\mu^{FL}(\d,(1+\e)\r)$ defined in \eqref{eq:mu_epsilon2}.  
If $\r_{\e}(\d)$ is the solution to
$\mu^{FL}(\d,(1+\e)\r)=1$, then for  $\r<\r_{\e}(\d)$ we achieve the desired bound, $\mu^{FL}(k,n,N)<1$,
to ensure $\sl1$.  The statement of Theorem
\ref{thm:LF1_asymptotic} follows from relating
$\r_{\e}(\d)$ to $\rfl$, the solution to $\mu^{FL}(\d,\r)=1$.

\begin{lemma}\label{lem:mu_epsilon2}
Fix $\epsilon>0$.  Consider the proportional-growth asymptotic with
parameters $(\d,\r)\in(0,1)\times(0,1/2)$.  Draw the corresponding $n\times
N$ matrices $A$ from the Gaussian ensemble.  Then
\begin{equation}\label{eq:mu_epsilon2}
Prob\left( \mu^{FL}(k,n,N)<\mu^{FL}(\d,(1+\e)\r)  \right)\rightarrow 1
\end{equation}
exponentially in $n$.
\end{lemma}

\begin{proof}
Theorem \ref{thm:LUbounds} and the form of $\mu^{FL}(\d,\r)$ imply a
similar bound to the above with a modified dependence on $\e$.
For any $c\e>0$, with $n/N\rightarrow\d\in (0,1)$ and
$k/n\rightarrow\r\in (0,1/2]$, the probability, on the draw of $A$ from the Gaussian 
ensemble, that 
\begin{equation}\label{eq:mu_epsilon_tmp}
\mu^{FL}(k,n,N)<\frac{1+\sqrt{2}}{4}\left(\frac{1+\mathcal{U}(\d,2\r)+c\e}{1-\mathcal{L}(\d,2\r)-c\e} -1 \right).
\end{equation}
is satisfied converges to one exponentially with $n$.
Since $\mathcal{U}(\d,\r)$ is non-decreasing in $\r$ and $\mathcal{L}(\d,\r)$ is strictly
increasing in $\r$ for any $\d$ and $\r\in(0,1)$, it follows that the 
right-hand side
of \eqref{eq:mu_epsilon_tmp} can be bounded by the right-hand side of 
\eqref{eq:mu_epsilon2} for any fixed $\e$ satisfying
$0<\e<\frac{1}{2\r}-1$, by setting
\[
c:=\frac{\r}{2}\left.\frac{\partial \mathcal{L}(\d,z)}{\partial z}\right|_{z=2(1+\e)\r}>0.
\]
(The upper bound on $\e$ is imposed so that the 
second argument of $\mathcal{U}(\d,\cdot)$ and $\mathcal{L}(\d,\cdot)$, $2(1+\e)\r$, is in the admissible 
range of $(0,1)$.)
That the bound \eqref{eq:mu_epsilon2} is satisfied for all $\e>0$
sufficiently small, and that the right hand side of
\eqref{eq:mu_epsilon2} is strictly increasing in $\e$ establishes that 
\eqref{eq:mu_epsilon2} is in fact satisfied probability on the draw of $A$ 
that converges to one exponentially in $n$ for {\em any} $\e\in\left(0,\frac{1}{2\r}-1
\right)$.  
\end{proof}

\begin{proof}[Theorem \ref{thm:LF1_asymptotic}]
Let $\r_{\e}(\d)$ be the solution of
$\mu^{FL}(\d,(1+\e)\r)=1$. Then, for any $\r<\r_{\e}(\d)$, Lemma
\ref{lem:mu_epsilon2} implies that $\mu^{FL}(k,n,N)<1$, which by Theorem 
\ref{thm:LF1}, ensures $\sl1$.  To remove the dependence
on the level curve $\r_{\e}(\d)$, note that
$\r_{\e}(\d)$ is related to $\rfl$, the solution
of $\mu^{FL}(\d,\r)=1$, by $(1+\e)\r_{\e}(\d)\equiv\rfl$.
Since $(1-\e)< (1+\e)^{-1}$ for all $\e>0$, we have
$(1-\e)\rfl<\r_{\e}(\d)$. Thus, provided $\r<(1-\e)\rfl$, the statement of
Theorem \ref{thm:LF1_asymptotic} is satisfied.
\end{proof}

\subsection{Comparison of bounds on $\sl1$}\label{subsec:comparison}

\begin{figure}[t]
 \begin{center}
 \includegraphics[bb= 70 215 546 589, width=2.70 in,height=2.00 in]{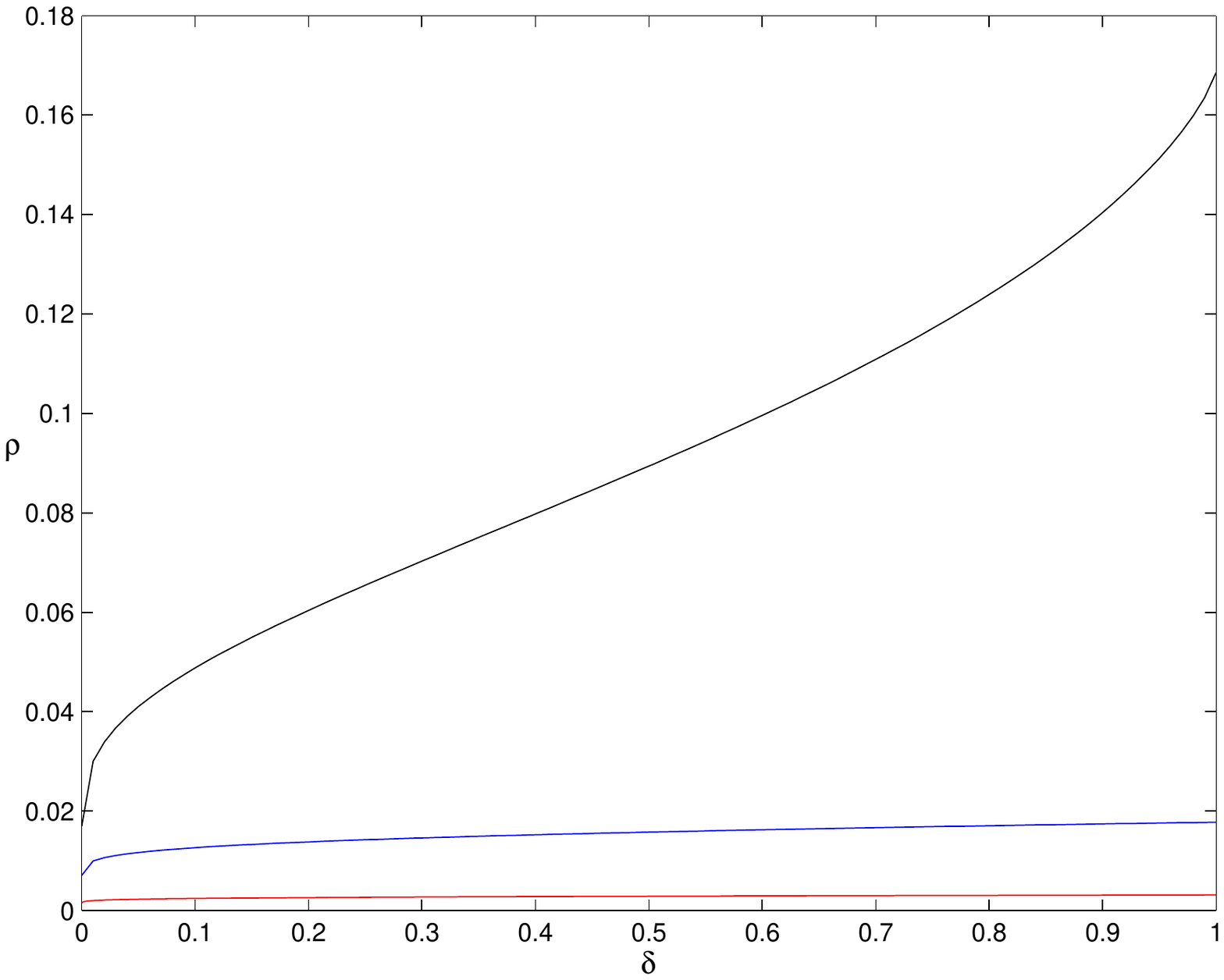}
\hspace*{0.5cm}
\includegraphics[bb= 70 215 546 589, width=2.70 in,height=2.00 in]{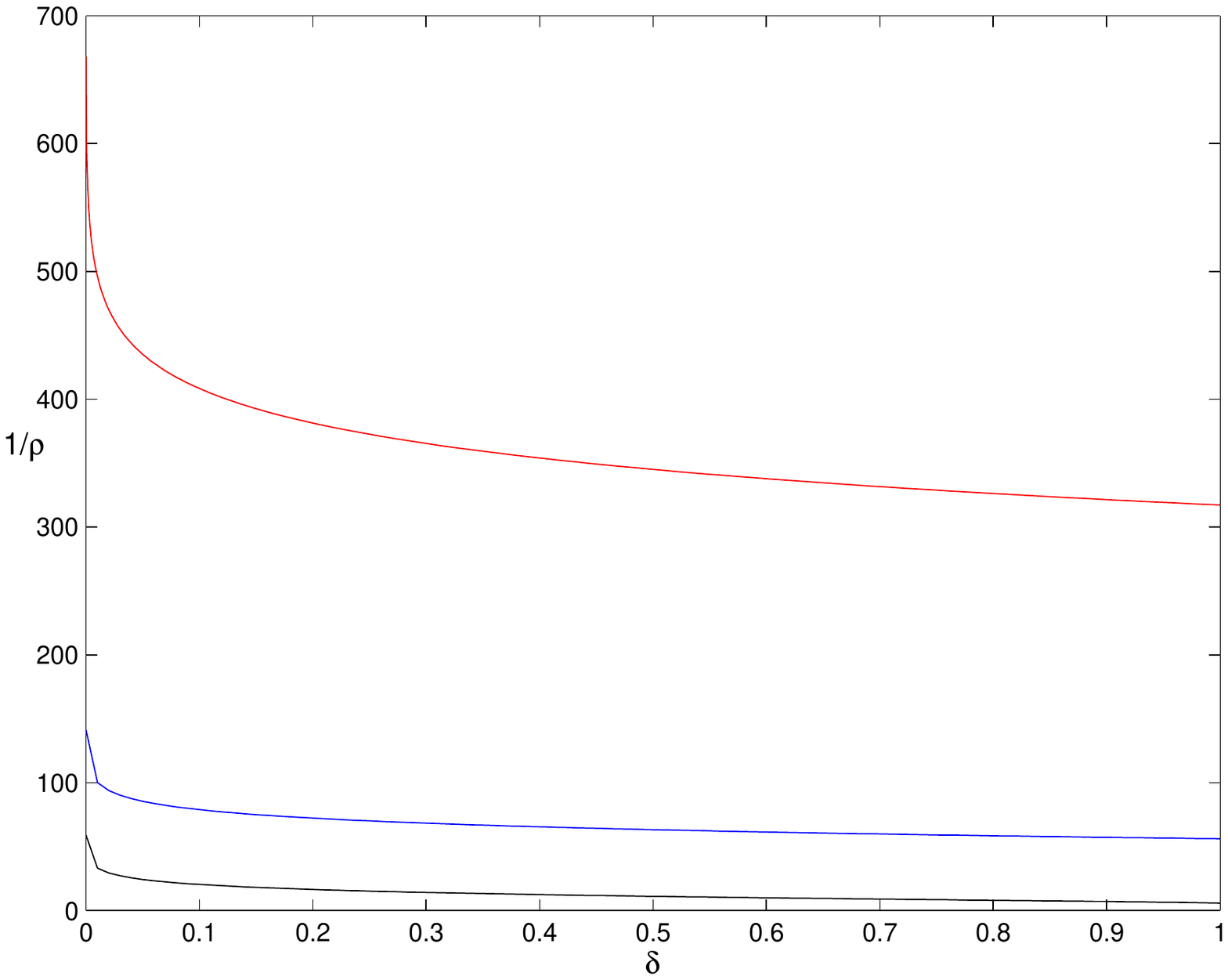}
 \caption{Left panel: Three lower bounds on the $\sl1$ phase
   transition, $\r_S(\d)$, for Gaussian random matrices from Theorem
   \ref{thm:D} ($\rd$, black), Theorem \ref{thm:RV} ($\rrv$, blue),
   and Theorem \ref{thm:LF1_asymptotic} ($\rfl$, red).  Right
   panel: The inverse of the $\sl1$ phase transition lower bounds in
 the left panel.\label{fig:l1}}
 \end{center}
 \end{figure}
 
In this section we use the phase transition framework to readily compare bounds on the region of $\sl1$ obtained from vastly different methods of analysis.  In Section \ref{subsec:l1l0}, we have already determined the region of strong equivalence for $\lone$-minimization obtained by using the RIP.  Here we look at two other examples, namely Donoho's polytope results \cite{Do05_signal, Do05_polytope} and the sufficient condition of Rudelson and Vershynin \cite{RV07_gaussian} obtained from geometric functional analysis.  We do not go into great details about how the results were obtained, but simply point out that the methods of analysis are rather different.  As a result, the original statements of the theorems take drastically different forms and are therefore  difficult to compare even qualitatively.  Translating the results into the phase transition framework,
however, offers a direct, quantitative, and simple method of comparison. 

Using polytope theory and the notion of central-neighborliness,  Donoho \cite{Do05_signal} defined a function $\rd$ which defines a region of the $(\d,\r)$ phase space ensuring StrongEquiv($A, \lone$) with probability on the draw of $A$ converging to one exponentially in $n$. The phase transition function $\rd$ is displayed as the black curve in Figure~\ref{fig:l1}.

\begin{theorem}[Donoho \cite{Do05_signal}]\label{thm:D}  
Fix $\epsilon>0$.   Consider the proportional-growth asymptotic, Definition \ref{def:pga}, with parameters $(\d,\r)\in(0,1)\times(0,1/2)$.  Sample each $n\times N$ matrix $A$ from the Gaussian ensemble.  Suppose $\r<(1-\e)\rd$.  Then $\hbox{Prob}(\hbox{StrongEquiv}(A,\lone))\rightarrow1$ as $n\rightarrow\infty$.
  
Therefore $\rd$ bounds a region of strong equivalence for $\lone$-minimization.
\end{theorem}

Rudelson and Vershynin \cite{RV07_gaussian} used an alternative
geometric approach from geometric functional analysis (GFA) to determine
regions of $\sl1$ for Gaussian and random partial Fourier matrices.
For Gaussian matrices their elegantly simple proof involves
employing Gordon's ``escape through the mesh theorem'' on the nullspace of
$A$.  Their lower bound on the region of $\sl1$ is larger for the
Gaussian ensemble than for the Fourier ensemble.  We restate their condition for the Gaussian ensemble in the proportional growth asymptotic.

\begin{definition}[GFA Region of $\sl1$]
Define
\[\gamma(\r\d):=\exp\left(\frac{\log(1+2log(e/\r\d))}{4\log(e/\r\d)}\right),\]
\begin{equation}\label{mu_rv}
\mu^{RV}(\d,\r):=\r\left(12+8\log(1/\r\d)\cdot\gamma^2(\r\d)\right),
\end{equation}
and $\rrv$ as the solution to $\mu^{RV}(\d,\r)=1$.
\end{definition}

The function $\rrv$ is displayed as the blue curve in Figure \ref{fig:l1}.  

\begin{theorem}[Rudelson and Vershynin \cite{RV07_gaussian}]\label{thm:RV}  
Fix $\epsilon>0$.   Consider the proportional-growth asymptotic, Definition \ref{def:pga}, with parameters $(\d,\r)\in(0,1)\times(0,1/2)$.  Sample each $n\times N$ matrix $A$ from the Gaussian ensemble.  Suppose $\r<(1-\e)\rrv$.  Then $\hbox{Prob}(\hbox{StrongEquiv}(A,\lone))\rightarrow1$ as $n\rightarrow\infty$.
  
Therefore $\rrv$ bounds a region of strong equivalence for $\lone$-minimization.\end{theorem}   

Versions of Theorems \ref{thm:D} and \ref{thm:RV} exist for finite values of $(k,n,N)$, 
\cite{DoTa08_finite,RV07_gaussian}, but in each case the recoverability conditions rapidly approach the stated asymptotic limiting functions $\rho_S(\d)$ as $(k,n,N)$ grow;
we do not further complicate the discussion with their rates of convergence.

Since Theorems \ref{thm:LF1_asymptotic}, \ref{thm:D}, and \ref{thm:RV} provide a region of $\sl1$, we now have three subsets of the exact region of $\sl1$.  Although Theorems \ref{thm:LF1_asymptotic}, \ref{thm:D}, and \ref{thm:RV} 
each have the same goal of quantifying the exact boundary of $\sl1$ for Gaussian random matrices, they are arrived at using substantially
different methods of analysis.  The efficacy of the bounds from the largest region of $\sl1$ to the smallest region are $\rd$ of Donoho, $\rrv$ of Rudelson and Vershynin, and $\rfl$ of Foucart and Lai, see the left panel of Figure \ref{fig:l1}.  From the inverse of $\rho_S(\d)$, see the right panel of Figure \ref{fig:l1}, we can read the constant of proportionality where the associated method of analysis guarantees $\sl1$: from Theorems \ref{thm:D}, \ref{thm:RV}, and
\ref{thm:LF1} they are bounded below by: $n\ge 5.9 k$, $n\ge 56 k$,
and $n\ge 317 k$ respectively.

\subsection{Further Considerations}

The phase transition framework can also be used to quantify what has been proven about an encoder/decoder pair's speed of convergence, its degree of robustness to noise, and to make comparisons of these properties between different algorithms.  A general framework for expressing the results of RIP based analyses as statements in the phase transition framework is presented in \cite{BCTT09}, where it is also applied to three exemplar greedy algorithms CoSaMP \cite{NeTr09_cosamp}, Subspace Pursuit \cite{SubspacePursuit}, and Iterated Hard Thresholding \cite{BlDa08_iht}.  Bounds on regions of $StrongEquiv(A,\ell^q)$ for $\ell^q$-minimization for $q\in(0,1]$ implied by the RIP are available in Section~\ref{subsec:lq}, where the effects of noise are also considered.  Through these ``objective'' measures of comparison we hope to make clear the proven efficacy of sparse approximation algorithms and allow for their transparent comparison.

In this article, we have considered only the case of noiseless measurements, Regions of Strong Equivalence, and a particular result obtained via an eigenvalue analysis and the RIP.  We briefly discuss some additional  considerations for the phase transition framework.

\subsubsection{Phase Transitions with Noisy Measurements}
In a practical setting, it is more reasonable to assume that the measurements are corrupted by noise, $y=Ax+e$ for some noise vector $e$.  The RIP has played a vital role in establishing stable signal recovery in the presence of noise for many decoders.  When noise is present, the curves $\rho_S(\delta)$ bounding regions of strong equivalence serve as an upper bound to the curves depicting the regions of the phase plane which guarantee stable recovery.  The RIP constants also describe how significantly the noise will be amplified by the encoder/decoder pairing, details are available for the Gaussian encoder and $\ell^q$-minimization decoder \cite{BCT_RIP_arxiv} and greedy decoders \cite{BCTT09}  CoSaMP, Subspace Pursuit, and Iterated Hard Thresholding. 
Hassibi and Xu have developed a stability analysis of  $\ell^1$-minimization
from the analysis of convex polytopes \cite{HassibiXu08}, establishing substantially
larger stability regions than the regions implied by the RIP.

\subsubsection{Regions of Weak Equivalence and Average Case Performance}
In many applications, it may not be imperative that the decoder is able to reconstruct every $k$-sparse vector.  Instead, one may be willing to lose a small fraction of all possible $k$-sparse signals.  This is the behavior observed when a decoder is tested on $k$-sparse vectors whose support sets are drawn uniformly at random.  Large scale empirical testing of CoSaMP, Subspace Pursuit and Iterated Hard Thresholding were compiled by Donoho and Maleki \cite{DoMa09}.  Most sparse approximation algorithms do not have a theoretical average case analysis.  The polytope analysis of Donoho and Tanner allows for analytical arguments providing a \emph{Region of Weak Equivalence} where recovery is guaranteed for all but a but a small fraction of $k$-sparse signals.  An average case variant of the RIP is being developed, see 
\cite{Tropp_sub}.

\subsubsection{Improving the RIP Phase Transition}
It is possible that Thm.~\ref{thm:LF1} could be improved with alternative methods of analysis.  For example, Thm.~\ref{thm:LF1} built off the work of Cand{\`e}s, Romberg, and Tao \cite{CandesRoot2,CRTStable,CTDecoding}.  In \cite{CandesRoot2}, Cand{\`e}s proved that if $R(2k,n,N)<\sqrt{2}-1$, then $\ell^1$-minimization will successfully recover every $k$-sparse vector.  An asymmetric analysis and translation into the Strong Equivalence terminology of Sec.~\ref{subsec:l1l0} produces a function $\rho_S^C(\delta)$ which bounds a region of strong equivalence.  The alternative methods of Foucart and Lai leading to Thm.~\ref{thm:LF1} provided a larger region of strong equivalence.  See Fig.~\ref{fig:improvements}.

\begin{figure}[t]
 \begin{center}
 \includegraphics[bb= 70 215 546 589, width=2.70 in,height=2.00 in]{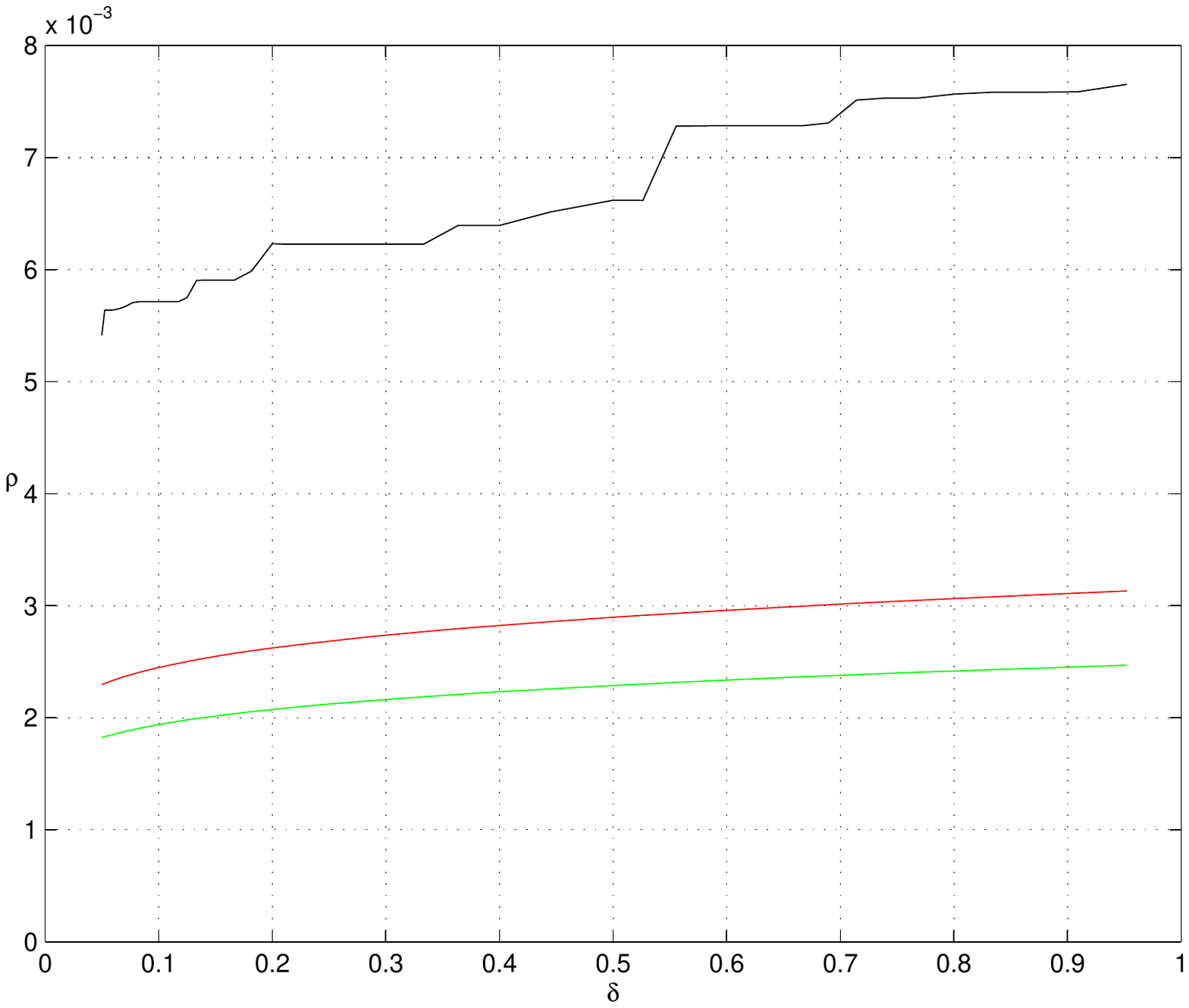}
%\hspace*{0.5cm}
%\includegraphics[bb= 70 215 546 589, width=2.70 in,height=2.00 in]{figure1b.pdf}
 \caption{Example Improvements on bounds on the $\sl1$ phase
   transition, $\r_S(\d)$, for Gaussian random matrices: ($\rho_S^C(\delta)$, green), Theorem \ref{thm:LF1_asymptotic} ($\rfl$, red), and ($\rho_S^{emp}(\delta)$, black).  \label{fig:improvements}}
 \end{center}
 \end{figure}

Alternatively, the region of strong equivalence might be increased by improving the bounds on the RIP constants, $\mathcal{L}(\delta,\rho),\mathcal{U}(\delta,\rho)$.  If the method of analysis remained the same, we can explore the effects of improved bounds by examining the statements with empirically observed {\em lower bounds} on RIP constants for Gaussian matrices.  As detailed in Table 
\ref{tab:emp}, the current bounds from Thm.~\ref{thm:LUbounds} are no more than twice the empirical RIP constants.  Replacing the RIP constants with empirically observed lower bounds of the RIP constants (for $n=800$) in $\mu^{FL}(k,n,N)$ gives us a function $\rho_S^{emp}(\delta)$, see Figure \ref{fig:improvements}, which is an upper bound on the region of strong equivalence implied by Thm.~\ref{thm:LF1}; this improvement is no more than 2.5 times $\rho_S^{FL}(\d)$ for $\d\in[1/20,20/21]$.

%% Additional Noisy Discussion

\section{$\ell^q$-regularization Phase Transitions for $q\in (0,1]$ Implied by RIP Constants}\label{subsec:lq}

Foucart and Lai improved on the previously best known RIP bounds of
Cand\`es ($q=1$) and Chartrand ($q\in (0,1)$) \cite{CandesRoot2,ChartrandExact} 
for $\ell^q$-regularization.  Theorem \ref{thm:LF1} is the simplest case
of Foucart and Lai's results, for $\ell^1$-regularization and $x$ exactly $k$
sparse.  More generally, they considered the family of $x_{\theta}$
which satisfy a scaled approximate fit to $b$,
\begin{equation}\label{eq:misfit}
\|Ax_{\theta}-b\|_2\le (1+U(2k,n,N))\cdot \theta
\quad\mbox{ for }\quad \theta\ge 0.
\end{equation}
Letting $\xs$ be the argmin for the
$\ell^q$-regularized constrained problem
\begin{equation}\label{eq:Pq}
\min_z\|z\|_q\quad\mbox{ subject to }\quad 
\|Az-b\|_2\le (1+U(2k,n,N))\cdot \theta,
\end{equation}
Foucart and Lai bounded the discrepancy between $\xs$ and any
$x_{\theta}$ satisfying (\ref{eq:misfit}) in terms of the discrepancy
between $x_{\theta}$ and its best $k$ sparse approximation, 
\begin{equation}\label{eq:sigma}
\sigma_k(\xt)_q:=\inf_{\|z\|_{\ell^0}\le k} \|\xt-z\|_q.
\end{equation}

\begin{theorem}[Foucart and Lai \cite{FoucartLai08}]  Given $q\in(0,1]$, for
  any matrix $A$ of size $n\times N$ with $n<N$ and with 
RIP constants $L(2k,n,N)$
  and $U(2k,n,N)$ and $\mu(2k,n,N)$ defined as \eqref{eq:mu_knN}, if
\begin{equation}\label{eq:alphamu}
\alpha^{1/2-1/q}\mu(2k\alpha,n,N)<1
\quad\mbox{ for any }\quad 1\le \alpha\le \frac{n}{2k},
\end{equation}
then a solution $\xs$ of (\ref{eq:Pq}) approximates any $x_{\theta}$ satisfying 
(\ref{eq:misfit}) within the bounds
\begin{eqnarray}
\|\xt-\xs\|_q & \le & C_1\cdot\sigma_k(\xt)_q\;+\; D_1\cdot
k^{1/q-1/2}\cdot\theta \\
\|\xt-\xs\|_2 & \le & C_2\cdot \sigma_k(\xt)_q\cdot
(\alpha k)^{1/2-1/q}\;+\; D_2\cdot\theta \label{eq:L2error}
\end{eqnarray}
with $C_1,C_2,D_1,$ and $D_2$ functions of $q$, $\alpha$, and 
$\frac{1+U(2\alpha k,n,N)}{1-L(2\alpha k,n,N)}$.   
\label{thm:LFq}
\end{theorem}

The parameter $\alpha$ in Theorem \ref{thm:LFq} is a free parameter 
from the method of proof, and should be selected so as to maximize 
the region where \eqref{eq:alphamu} and/or other conditions are satisfied.
For brevity we do not state the formulae for $C_1,C_2,D_1,$ and $D_2$ as 
functions of $(k,n,N)$, but only state them in Theorem 
\ref{thm:LFq_asymptotic} in terms of their  
bounds for Gaussian random matrices as $(k,n,N)\rightarrow\infty$.

Although the solution of \eqref{eq:Pq}, $\xs$, has unknown sparsity, 
Theorem \ref{thm:LFq} ensures that if there is a solution of 
\eqref{eq:misfit}, $\xt$, which can be well approximated by a $k$ sparse 
vector, i.e. if $\sigma_k(\xt)_q$ is small, then if \eqref{eq:alphamu} is 
satisfied the discrepancy between $\xs$ and $\xt$ will be similarly small.
For instance, if the sparsest solution of \eqref{eq:Pq}, $\xt$, is $k$ 
sparse, then \eqref{eq:L2error} implies that 
$\|\xt-\xs\|_2\le D_2\cdot \theta$; moreover, if $\theta=0$ then $\xs$ 
will be $k$ sparse and satisfy $A\xs=b$ (in the case $q=1$ this result 
is summarized as Theorem \ref{thm:LF1}).  Substituting bounds on 
the asymmetric RIP constants $L(2\alpha k,n,N)$ and 
$U(2\alpha k,n,N)$ from Theorem 
\ref{thm:LUbounds} we arrive at a quantitative version of Theorem 
\ref{thm:LFq} for Gaussian random matrices.

\begin{theorem} Given $q\in (0,1]$, for any $\epsilon>0$, as 
$(k,n,N)\rightarrow\infty$ with $n/N\rightarrow\delta\in(0,1)$ 
and $k/n\rightarrow\r\in(0,1/2]$, if $\r<(1-\epsilon)\rflq$ 
where $\rflq$ is the maximum over $1\le \alpha\le 1/2\r$ of the 
solutions, $\r(\d;q,\alpha)$, of
$\mu_{\alpha}(\d,2\alpha\r):=\alpha^{1/2-1/q}\mu(\d,2\alpha\r)=1$ with
$\mu(\d,2\r)$ defined as in \eqref{mu_dr},
there is overwhelming probability on the draw of $A$ with 
Gaussian i.i.d. entries that 
  a solution $\xs$ of (\ref{eq:Pq}) approximates any $\xt$ satisfying 
  (\ref{eq:misfit}) within the bounds
\begin{eqnarray}
\|\xt-\xs\|_q & \le & C_1(\d,2\alpha\r)\cdot\sigma_k(\xt)_q\;+\; D_1(\d,2\alpha\r)\cdot
k^{1/q-1/2}\cdot\theta \\
\|\xt-\xs\|_2 & \le & C_2(\d,2\alpha\r)\cdot \sigma_k(\xt)_q\cdot
(\alpha k)^{1/2-1/q}\;+\; D_2(\d,2\alpha\r)\cdot\theta.
\end{eqnarray}
The multiplicative ``stability factors'' are defined as:
\begin{eqnarray}
C_1(\d,2\alpha\r) & := & \frac{2^{2/q-1}(1+\mu_{\alpha}(\d,2\alpha\r)^q)^{1/q}}{(1-\mu_{\alpha}(\d,2\alpha\r)^q)^{1/q}}, \quad\quad
D_1(\d,2\alpha\r)  :=  \frac{2^{2/q-1}\beta(\d,2\alpha\r)}{(1-\mu_{\alpha}(\d,2\alpha\r)^q)^{1/q}} \nonumber \\
C_2(\d,2\alpha\r) & := & \frac{2^{2/q-2}(\beta(\d,2\alpha\r)+1-\sqrt{2})}{(1-\mu_{\alpha}(\d,2\alpha\r)^q)^{1/q}}, \nonumber \\
D_2(\d,2\alpha\r) & := & \frac{2^{1/q-2}\beta(\d,2\alpha\r)(\beta(\d,2\alpha\r)+1-\sqrt{2})}{(1-\mu_{\alpha}(\d,2\alpha\r)^q)^{1/q}}+2\beta(\d,2\alpha\r) 
\end{eqnarray}
with $\beta(\d,\r):=(1+\sqrt{2})\frac{1+\mathcal{U}(\d,\r)}{1-\mathcal{L}(\d,\r)}$.
\label{thm:LFq_asymptotic}
\end{theorem}

Unlike Theorem \ref{thm:LF1_asymptotic} which specifies one function $\rfl$ 
which bounds from below the phase transition for $\sl1$, Theorem \ref{thm:LFq_asymptotic} specifies a multiparameter 
family of threshold functions depending on $q$ and possibly with 
further dependence on bounds on the 
multiplicative stability factors, such as $C_1(\d,\r)$.  The  
function $\rfl$ in Theorem \ref{thm:LF1_asymptotic} 
corresponds to the case $\theta=0$, $q=1$, and no bounds on the 
stability parameters.  The function $\rflq$ in Theorem 
\ref{thm:LFq_asymptotic} corresponds to $\ell^q$ regularization  with 
unbounded stability coefficients, and as a result, it is only meaningful 
for $\rho$ strictly below $\rflq$ or in the case where 
there exists a $k$ sparse solution to $A\xt=b$ and $\theta=0$.
More generally, specifying a bound on one or more of the multiplicative 
stability factors  determines functions $\rflqcond$. For instance, 
imposing a bound of $\Upsilon$ on the stability factor $C_1(\d,\r)$ 
generates a function $\rflqc1$; Figure \ref{fig:Lq} shows $\rflqc1$ 
for $q=1$ and $q=1/2$ in panels (a-b) and (c-d) respectively.  
Software is available upon request which will generate  
functions with these and other choices of parameters in Theorem 
\ref{thm:LFq_asymptotic}.

\begin{figure}[t]
 \begin{center}
\begin{tabular}{cc}
 \includegraphics[bb= 70 185 546 589, width=2.70 in,height=2.00 in]{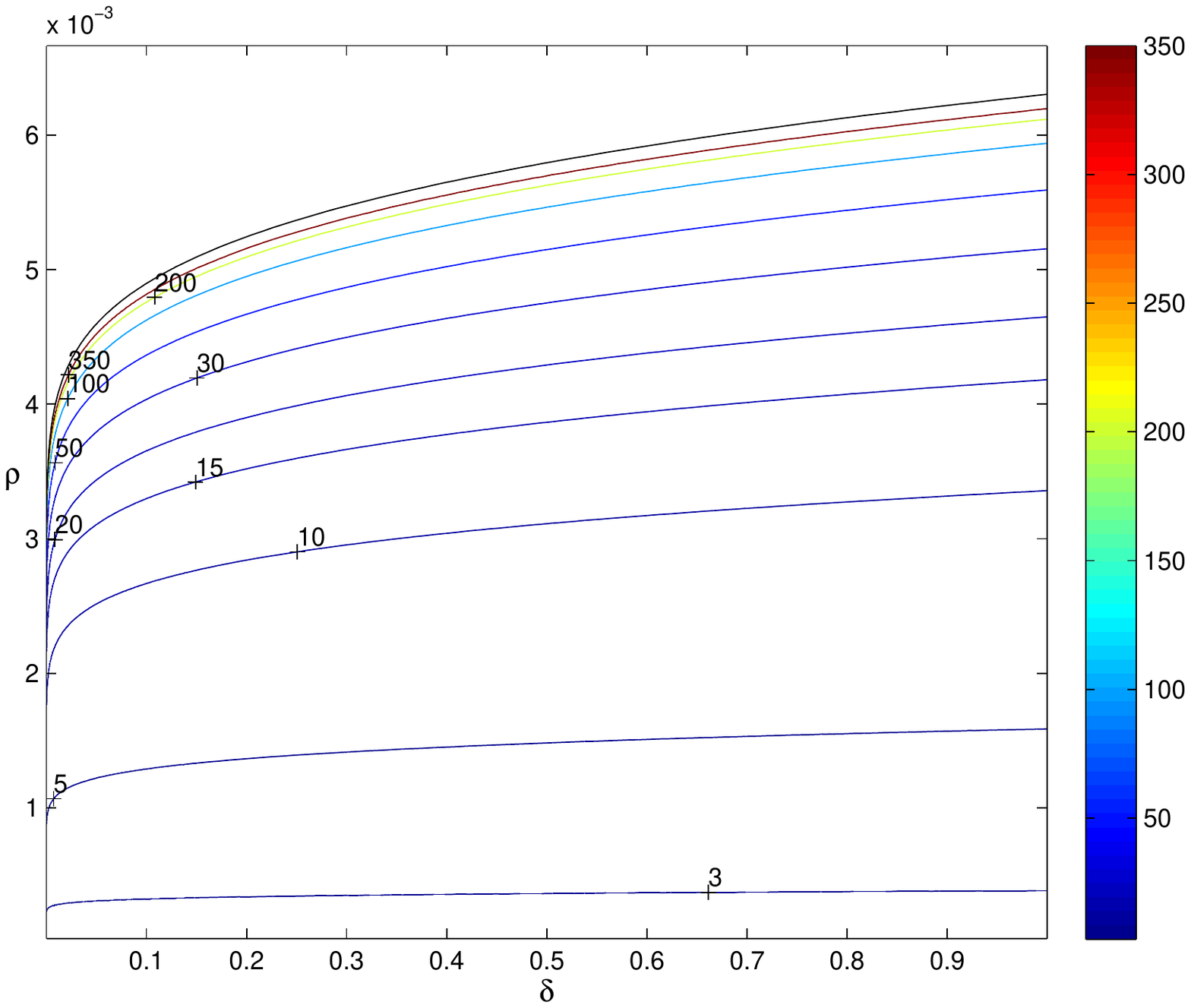} &
\includegraphics[bb= 70 185 546 589, width=2.70 in,height=2.00 in]{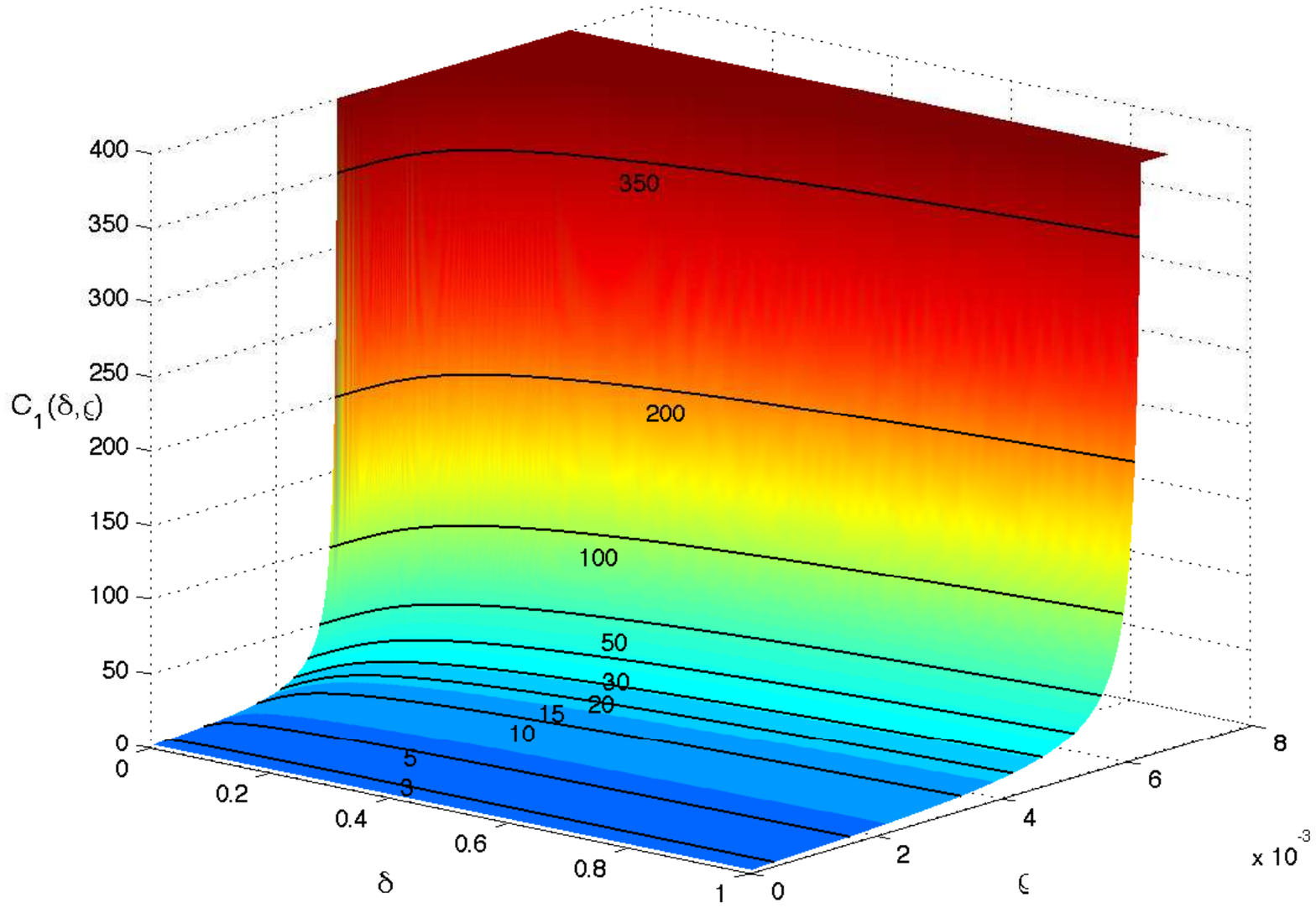} \\
(a) & (b) \\
\includegraphics[bb= 70 185 546 619, width=2.70 in,height=2.00 in]{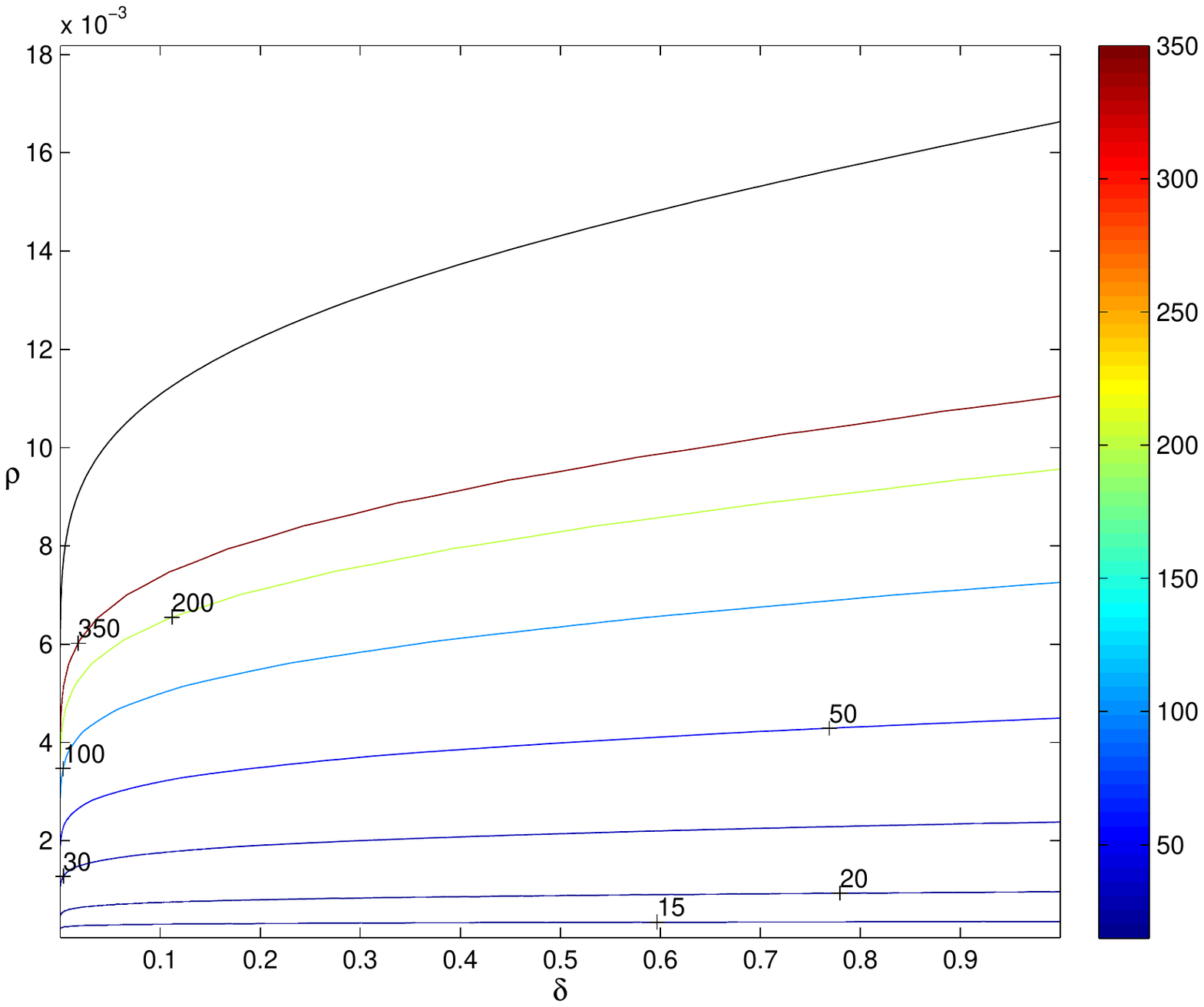} &
\includegraphics[bb= 70 185 546 619, width=2.70 in,height=2.00 in]{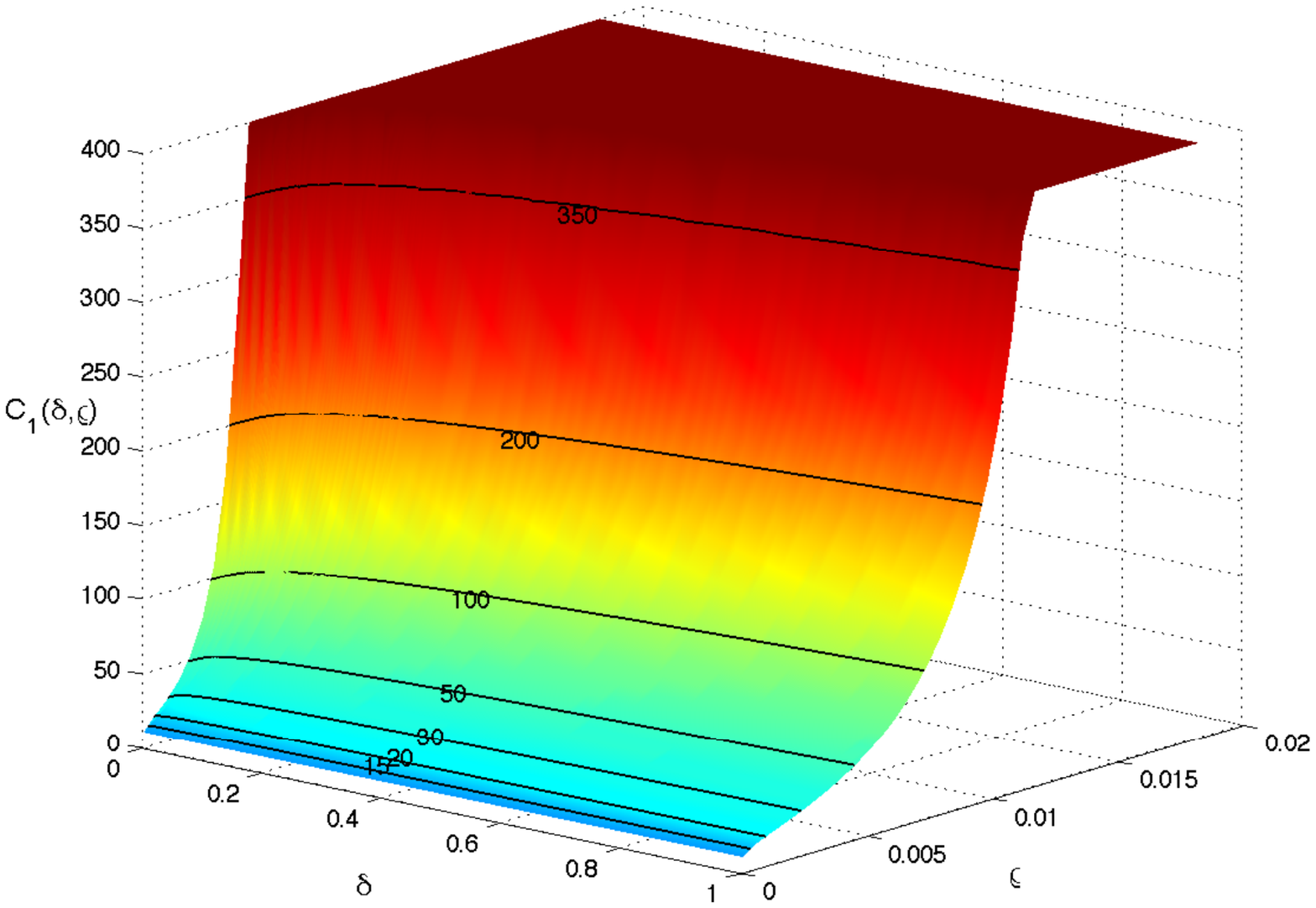} \\
(c) & (d) \\
\end{tabular}
 \caption{The surface whose level curves specify $\rflqc1$ for $q=1$
   and $q=1/2$ in Panels (b) and (d) respectively, with level curves
   for specific values of $\Upsilon$ shown in Panels (a) and (c) for
   $q=1$ and $q=1/2$ respectively}\label{fig:Lq}
 \end{center}
 \end{figure}

\subsection{Discussion}\label{sec:Discussion}

The lower bound on the $\sl1$ phase transition implied by the RIP
for strictly $k$ sparse signals, $\rfl$ of Theorem \ref{thm:LF1},
does not have any implied stability.  In order to ensure stability,
Theorem \ref{thm:LF1} requires further restrictive bounds on the {\em
  stability factors} in Theorem \ref{thm:LFq}, further reducing the
lower bound on the phase transition.  
For example, $C_1(\d,\r)$ for $q=1$ is shown in Figure \ref{fig:Lq} (b), with 
level curves of $C_1(\d,\r)$ corresponding to fixed stability factors
proceeding \eqref{eq:sigma} in Theorem \ref{thm:LFq}; phase
transitions below which specified bounds on $C_1(\d,\r)$ can be
ensured are shown in Figure \ref{fig:Lq} (a-b).  
The stability factor becomes unbounded at finite $\r$ as $\r\uparrow\rfl$.

A current trend in sparse approximation is to consider
$\ell^q$-regularization for $q\in(0,1)$, with the aim of increasing
the recoverability region \cite{ChartrandExact,ChartrandRIP}.  
Existing results have shown that indeed 
the region where $\ell^q$-regularization successfully recovers $k$
sparse vectors at least does not decrease as $q$ decreases
\cite{Gribonval}, though little is known as to the rate, if any, at which it 
increases.  Theorem \ref{thm:LFq} gives lower bounds on these regions where
$\ell^q$-regularization is guaranteed to have specified recoverability
properties, and in fact for any strictly $k$ sparse vector it implies
that if \eqref{eq:mu_knN} is finite, there is a small enough $q$ such
every $k$ sparse vector can be recovered exactly from $(b,A)$ by solving
\eqref{eq:Pq} with $\theta=0$.  Despite this and other encouraging
results, many fundamental questions about $\ell^q$-regularization
remain, in particular how to find the global minimizer of \eqref{eq:Pq}.
Moreover, it is unknown if $\ell^q$-regularization remains stable as
$q$ decreases.  In order to ensure stability, Theorem \ref{thm:LF1}
requires further restrictive bound on the stability factor
in Theorem \ref{thm:LFq}, further reducing the lower bound on the phase 
transition.  For example, $C_1(\d,\r)$ for $q=1/2$ is shown in Figure
\ref{fig:Lq} (d), with level curves of $C_1(\d,\r)$ corresponding to
fixed stability factors preceeding \eqref{eq:sigma} in Theorem
\ref{thm:LFq}; phase transitions below which specified bounds on
$C_1(\d,\r)$ can be ensured are shown in Figure \ref{fig:Lq} (c-d).
Decreasing $q$ from 1 to $1/2$ does increase the value of $\r$ at
which the stability factors in Theorem \ref{thm:LFq} become unbounded;
however, comparing Figure \ref{fig:Lq} (d) and (b) it is
apparent that this elevating of the unstable phase transition comes at
the price of also elevating $C_1(\d,\r)$ for small values of $\r$.
In particular, the region where $C_1(\d,\r)\le 50$ is, in fact, larger
for $q=1$ than for $q=1/2$.

{\bf Acknowledgements.}\quad 
The authors would like to thank the editor and the referees for their useful suggestions 
that have greatly improved the manuscript.

\bibliographystyle{plain}
\bibliography{rip_noisy}

\begin{thebibliography}{10}

\bibitem{BCT_RIP_arxiv}
J.~D. Blanchard, C.~Cartis, and J.~Tanner.
\newblock Compressed sensing: how sharp is the restricted isometry property.
\newblock extended technical report, 2009.

\bibitem{BCT09_DecayRIP}
J.~D. Blanchard, C.~Cartis, and J.~Tanner.
\newblock Decay properties for restricted isometry constants.
\newblock {\em IEEE Signal Proc. Letters}, 16(7):572--575, 2009.

\bibitem{BCTT09}
J.~D. Blanchard, C.~Cartis, J.~Tanner, and A.~Thompson.
\newblock Phase transitions for greedy sparse approximation algorithms.
\newblock submitted, 2009.

\bibitem{BlDa08_iht}
T.~Blumensath and M.~E. Davies.
\newblock Iterative hard thresholding for compressed sensing.
\newblock {\em Appl. Comp. Harm. Anal.}, 27(3):265--274, 2009.

\bibitem{CS_SIREV}
A.~M. Bruckstein, David~L. Donoho, and Michael Elad.
\newblock From sparse solutions of systems of equations to sparse modeling of
  signals and images.
\newblock {\em SIAM Review}, 51(1):34--81, 2009.

\bibitem{CompressiveSampling}
E.~J. Cand{\`e}s.
\newblock Compressive sampling.
\newblock In {\em International Congress of Mathematicians. Vol. III}, pages
  1433--1452. Eur. Math. Soc., Z\"urich, 2006.

\bibitem{CandesRoot2}
E.~J. Cand{\`e}s.
\newblock The restricted isometry property and its implications for compressed
  sensing.
\newblock {\em C. R. Math. Acad. Sci. Paris}, 346(9-10):589--592, 2008.

\bibitem{CRTRobust}
E.~J. Cand{\`e}s, J.~Romberg, and T.~Tao.
\newblock Robust uncertainty principles: exact signal reconstruction from
  highly incomplete frequency information.
\newblock {\em IEEE Trans. Inform. Theory}, 52(2):489--509, 2006.

\bibitem{CRTStable}
E.~J. Cand{\`e}s, J.~Romberg, and T.~Tao.
\newblock Stable signal recovery from incomplete and inaccurate measurements.
\newblock {\em Comm. Pure Appl. Math.}, 59(8):1207--1223, 2006.

\bibitem{CTDecoding}
E.~J. Cand{\`e}s and T.~Tao.
\newblock Decoding by linear programming.
\newblock {\em IEEE Trans. Inform. Theory}, 51(12):4203--4215, 2005.

\bibitem{CTNearOptimal}
E.~J. Cand{\`e}s and T.~Tao.
\newblock Near-optimal signal recovery from random projections: universal
  encoding strategies?
\newblock {\em IEEE Trans. Inform. Theory}, 52(12):5406--5425, 2006.

\bibitem{ChartrandExact}
R.~Chartrand.
\newblock Exact reconstructions of sparse signals via nonconvex minimization.
\newblock {\em IEEE Signal Process. Lett.}, 14:707--710, 2007.

\bibitem{ChartrandRIP}
R.~Chartrand and V.~Staneva.
\newblock Restricted isometry properties and nonconvex compressive sensing.
\newblock {\em Inverse Problems}, 24(035020):1--14, 2008.

\bibitem{CDS98}
S.~S. Chen, D.~L. Donoho, and M.~A. Saunders.
\newblock Atomic decomposition by basis pursuit.
\newblock {\em SIAM Rev.}, 43(1):129--159 (electronic), 2001.
\newblock Reprinted from SIAM J. Sci. Comput. {{\bf{2}}0} (1998), no. 1,
  33--61.

\bibitem{SubspacePursuit}
W.~Dai and O.~Milenkovic.
\newblock Subspace pursuit for compressive sensing signal reconstruction.
\newblock {\em IEEE Trans. Inform. Theory}, 55(5):2230--2249, 2009.

\bibitem{Do05_signal}
D.~L. Donoho.
\newblock Neighborly polytopes and sparse solution of underdetermined linear
  equations.
\newblock Technical Report, Department of Statistics, Stanford University,
  2005.

\bibitem{CompressedSensing}
D.~L. Donoho.
\newblock Compressed sensing.
\newblock {\em IEEE Trans. Inform. Theory}, 52(4):1289--1306, 2006.

\bibitem{Do05_polytope}
D.~L. Donoho.
\newblock High-dimensional centrally symmetric polytopes with neighborliness
  proportional to dimension.
\newblock {\em Discrete Comput. Geom.}, 35(4):617--652, 2006.

\bibitem{DoMa09}
D.~L. Donoho and A.~Maleki.
\newblock Optimally tuned iterative thresholding algorithms for compressed
  sensing.
\newblock {\em IEEE Sel. Topics Signal Processing}, in press.

\bibitem{DoSt06_breakdown}
D.~L. Donoho and V.~Stodden.
\newblock Breakdown point of model selection when the number of variables
  exceeds the number of observations.
\newblock In {\em Proceedings of the International Joint Conference on Neural
  Networks}, 2006.

\bibitem{DoTa05_signal}
D.~L. Donoho and J.~Tanner.
\newblock Sparse nonnegative solutions of underdetermined linear equations by
  linear programming.
\newblock {\em Proc. Natl. Acad. Sci. USA}, 102(27):9446--9451, 2005.

\bibitem{DoTa08_JAMS}
D.~L. Donoho and J.~Tanner.
\newblock Counting faces of randomly projected polytopes when the projection
  radically lowers dimension.
\newblock {\em J. AMS}, 22(1):1--53, 2009.

\bibitem{DoTa08_finite}
D.~L. Donoho and J.~Tanner.
\newblock Exponential bounds implying construction of compressed sensing
  matrices, error-correcting codes and neighborly polytopes by random sampling.
\newblock {\em IEEE Trans. on Information Theory}, 2010.
\newblock in press.

\bibitem{DoTs06_fast}
D.~L. Donoho and Y.~Tsaig.
\newblock Fast solution of l1 minimization problems when the solution may be
  sparse.
\newblock {\em IEEE Trans. Inform. Theory}, 54(11):4789--4812, 2008.

\bibitem{DTDS08}
D.~L. Donoho, Y.~Tsaig, I.~Drori, and J.-L. Stark.
\newblock Sparse solution of underdetermined linear equations by stagewise
  orthogonal matching pursuit.
\newblock {\em IEEE Trans. Inform. Theory}, submitted.

\bibitem{Dor43}
R.~Dorfman.
\newblock The detection of defective members of large populations.
\newblock {\em Ann. Math. Statist.}, 14(4):436--440, 1943.

\bibitem{Dossal}
C.~Dossal, G.~Peyr\'{e}, and J.~Fadili.
\newblock A numerical exploration of compressed sampling recovery.
\newblock {\em Linear Algebra Appl.}, 432(7):1663--1679, 2010.

\bibitem{CS_pixel}
M.~F. Duarte, M.~A. Davenport, D.~Takhar, Laska~J. N., T.~Sun, K.~F. Kelly, and
  R.~G. Baraniuk.
\newblock Single-pixel imaging via compressed sampling.
\newblock {\em IEEE Signal Processing Magazine}, 25(2):83--91, 2008.

\bibitem{EdelmanEigenvalues88}
A.~Edelman.
\newblock Eigenvalues and condition numbers of random matrices.
\newblock {\em SIAM J. Matrix Anal. Appl.}, 9(4):543--560, 1988.

\bibitem{Edelman_acta}
A.~Edelman and N.~R. Rao.
\newblock Random matrix theory.
\newblock {\em Acta Numer.}, 14:233--297, 2005.

\bibitem{GPSR}
M.~A.~T. Figueiredo, R.~D. Nowak, and S.~J. Wright.
\newblock Gradient projection for sparse reconstruction: Application to
  compressed sensing and other inverse problems.
\newblock {\em IEEE J. Sel. Topics Signal Process.}, 1(4):586--597, 2007.

\bibitem{FoucartLai08}
S.~Foucart and M.-J. Lai.
\newblock Sparsest solutions of underdetermined linear systems via
  $\ell_q$-minimization for $0 < q \le 1$.
\newblock {\em Appl. Comput. Harmon. Anal.}, 26(3):395--407, 2009.

\bibitem{Geman1980}
S.~Geman.
\newblock A limit theorem for the norm of random matrices.
\newblock {\em Ann. Probab.}, 8(2):252--261, 1980.

\bibitem{group_cs}
A.~C. Gilbert, M.~A. Iwen, and M.~J. Strauss.
\newblock Group testing and sparse signal recovery.
\newblock In {\em 42nd Asilomar Conference on Signals, Systems, and Computers},
  2008.

\bibitem{Gribonval}
R.~Gribonval and M.~Nielsen.
\newblock Sparse representations in unions of bases.
\newblock {\em IEEE Trans. Inform. Theory}, 49(12):3320--3325, 2003.

\bibitem{HartmannRieger}
A.~K. Hartmann and H.~Rieger.
\newblock {\em New Optimization Algorithms in Physics}.
\newblock Wiley VCH, Cambridge, 2006.

\bibitem{HartmannWeight}
A.~K. Hartmann and M.~Weight.
\newblock {\em Phase Transitions in Combinatorial Optimization Problems}.
\newblock Wiley VCH, Cambridge, 2005.

\bibitem{Richtarik}
M.~Journ\'{e}e, Y.~Nesterov, P.~Richt\'{a}rik, and R.~Sepulchre.
\newblock Generalized power method for sparse principal component analysis.
\newblock {\em Journal of Machine Learning Research}, 11:451--487, 2010.

\bibitem{MRI_fast}
M.~Lustig, D.~L. Donoho, and J.~M. Pauly.
\newblock Sparse {MRI}: The application of compressed sensing for rapid mr
  imaging.
\newblock {\em Magnetic Resonance in Medicine}, 58(6):1182--1195, 2007.

\bibitem{CS_MRI}
M.~Lustig, D.~L. Donoho, J.~M. Santos, and J.~M. Pauly.
\newblock Compressed sensing {MRI}.
\newblock {\em IEEE Signal Processing Magazine}, 25(2):72--82, 2008.

\bibitem{mehmet}
M.~Hedjazi Moghari, M.~Ak?akaya, A.~O'Connor, P.~Hu, V.~Tarokh, W.~J. Manning,
  and R.~Nezafat.
\newblock {C}o{SM}o: {C}ompressed sensing motion correction for coronary {MRI}.
\newblock The Annual Scientific Meeting of International Society for Magnetic
  Resonance in Medicine (ISMRM), 2010.

\bibitem{NPhard}
B.~K. Natarajan.
\newblock Sparse approximate solutions to linear systems.
\newblock {\em SIAM J. Comput.}, 24(2):227--234, 1995.

\bibitem{NeTr09_cosamp}
D.~Needell and J.~Tropp.
\newblock Cosamp: Iterative signal recovery from incomplete and inaccurate
  samples.
\newblock {\em Appl. Comp. Harm. Anal.}, 26(3):301--321, 2009.

\bibitem{NeVe06_UUP}
D.~Needell and R.~Vershynin.
\newblock Uniform uncertainty principle and signal recovery via regularized
  orthogonal matching pursuit.
\newblock {\em Foundations of Comp. Math.}, 9(3):317--334, 2009.

\bibitem{Healy_radar}
V.~M. Patel, G.~R. Easley, D.~M.~Jr Healy, and R.~Chellappa.
\newblock Compressed synthetic aperture radar.
\newblock {\em IEEE Sel. Topics Signal Processing}, in press.

\bibitem{RV07_gaussian}
M.~Rudelson and R.~Vershynin.
\newblock On sparse reconstruction from {F}ourier and {G}aussian measurements.
\newblock {\em Comm. Pure Appl. Math.}, 61(8):1025--1045, 2008.

\bibitem{Yilmaz}
R.~Saab and O.~Yilmaz.
\newblock Sparse recovery by non-convex optimization--instance optimality.
\newblock {\em Appl. Comp. Harm. Anal.}, in press.

\bibitem{Silverstein1985}
J.~W. Silverstein.
\newblock The smallest eigenvalue of a large-dimensional {W}ishart matrix.
\newblock {\em Ann. Probab.}, 13(4):1364--1368, 1985.

\bibitem{Tropp_sub}
J.~Tropp.
\newblock On the conditioning of random subdictionaries.
\newblock {\em Appl. Comp. Harm. Anal.}, 25(1):1--24, 2008.

\bibitem{Greed}
J.~A. Tropp.
\newblock Greed is good: algorithmic results for sparse approximation.
\newblock {\em IEEE Trans. Inform. Theory}, 50(10):2231--2242, 2004.

\bibitem{BeFr08}
E.~van~den Berg and M.~P. Friedlander.
\newblock Probing the pareto frontier for basis pursuit solutions.
\newblock {\em SIAM Journal on Scientific Computing}, 31(2):890--912, 2008.

\bibitem{WhittakerWatson}
E.~T. Whittaker and G.~N. Watson.
\newblock {\em A course of modern analysis}.
\newblock Cambridge Mathematical Library. Cambridge University Press,
  Cambridge, 1996.
\newblock An introduction to the general theory of infinite processes and of
  analytic functions; with an account of the principal transcendental
  functions, Reprint of the fourth (1927) edition.

\bibitem{HassibiXu08}
W.~Xu and B.~Hassibi.
\newblock Compressed sensing over the grassmann manifold: A unified analytical
  framework.
\newblock Forty-Sixth Annual Allerton Conference, 2008.

\bibitem{Bregman}
W.~Yin, S.~Osher, D.~Goldfarb, and J.~Darbon.
\newblock Bregman iterative algorithms for $\ell^1$-minimization with
  applications to compressed sensing.
\newblock {\em SIAM Journal on Imaging Science}, 1(1):143--168, 2008.

\end{thebibliography}

\end{document}